\documentclass[journal,draftclsnofoot,onecolumn,a4paper,12pt]{IEEEtran}

\pdfoutput=1

\usepackage[utf8]{inputenc}

\usepackage[T1]{fontenc}

\usepackage[english]{babel}

\usepackage{csquotes}

\usepackage{paralist}

\usepackage{tikz}
\usepackage{pgfplots}
\usepackage{pgfplotstable}
\usetikzlibrary{external}
\usetikzlibrary{calc,shapes,fit,matrix,chains,scopes,positioning}
\usepackage{amsmath}

\usepackage{amsthm}

\usepackage{glossaries}

\usepackage{booktabs}

\usepackage{scalefnt}

\usepackage{subcaption}

\usepackage{ifthen}

\usepackage{xspace}

\usepackage{color}

\usepackage{array}

\usepackage{tabularx}

\usepackage{verbatim}

\usepackage{soul}

\usepackage{calc}

\usepackage{flushend}

\providecommand*{\LayoutLib}{UNDEFINED}
\ifthenelse{\equal{\LayoutLib}{UNDEFINED}}{
\renewcommand*{\LayoutLib}{DEFINED}

\iflanguage{english}{%

}{}

\iflanguage{german}{%

}{}

\newcommand*{\TheoremName}{Theorem}
\newcommand*{\LemmaName}{Lemma}
\newcommand*{\CorollaryName}{Corollary}
\newcommand*{\DefinitionName}{Definition}
\newcommand*{\PremiseName}{Premise}
\newcommand*{\RemarkName}{Remark}
\newcommand*{\ExampleName}{Example}
\newcommand*{\ExerciseName}{Exercise}
\newcommand*{\ModelName}{Model}
\newcommand*{\NotationName}{Notation}
\newcommand*{\CaseName}{Case}
\newcommand*{\QuestionName}{Question}
\newcommand*{\AssumptionName}{Assumption}
\newcommand*{\NoteName}{Note}
\newcommand*{\ProblemName}{Problem}
\newcommand*{\SolutionName}{solution}
\theoremstyle{plain}
\newtheorem{theorem}{\TheoremName}
\newtheorem*{theorem*}{\TheoremName}

\newtheorem*{lemma*}{\LemmaName}
\newtheorem{corollary}{\CorollaryName}
\newtheorem*{corollary*}{\CorollaryName}

\theoremstyle{definition}

\newtheorem*{definition*}{\DefinitionName}

\newtheorem*{remark*}{\RemarkName}

\newtheorem*{example*}{\ExampleName}

\newtheorem*{exercise*}{\ExerciseName}

\newtheorem*{model*}{\ModelName}

\newtheorem*{notation*}{\NotationName}

\newtheorem*{premise*}{\PremiseName}
\theoremstyle{remark}

\newtheorem*{case*}{\CaseName}

\newtheorem*{question*}{\QuestionName}

\newtheorem*{assumption*}{\AssumptionName}

\newtheorem*{note*}{\NoteName}

\newtheorem*{problem*}{\ProblemName}

\newtheorem*{solution*}{\SolutionName}

\newcolumntype{P}[2]{>{$\rlap{\rule{#2}{0pt}}}#1<{$}}

\newcolumntype{M}[1]{>{$}#1<{$}}

\newcolumntype{T}{>{\small\tt}l}

\newcolumntype{E}{>{\small}X}

\definecolor{GrayOOO}{gray}{0.9}
\definecolor{GrayOOI}{gray}{0.8}
\definecolor{GrayOIO}{gray}{0.7}
\definecolor{GrayOII}{gray}{0.6}
\definecolor{GrayIOO}{gray}{0.5}
\definecolor{GrayIOI}{gray}{0.4}
\definecolor{GrayIIO}{gray}{0.3}
\definecolor{GrayIII}{gray}{0.2}

\definecolor{RedOOO}{rgb}{0.9,0,0}
\definecolor{RedOOI}{rgb}{0.8,0,0}
\definecolor{RedOIO}{rgb}{0.7,0,0}
\definecolor{RedOII}{rgb}{0.6,0,0}
\definecolor{RedIOO}{rgb}{0.5,0,0}
\definecolor{RedIOI}{rgb}{0.4,0,0}
\definecolor{RedIIO}{rgb}{0.3,0,0}
\definecolor{RedIII}{rgb}{0.2,0,0}

\definecolor{BlueOOO}{rgb}{0.2,0.2,0.9}
\definecolor{BlueOOI}{rgb}{0.2,0.2,0.8}
\definecolor{BlueOIO}{rgb}{0.2,0.2,0.7}

}{} %* \ifthenelse at the beginning of this file (don't remove the brackets) ***

\usepackage{ifthen}

\usepackage{amsmath}

\usepackage{amssymb}

\usepackage{eucal}

\usepackage{exscale}

\usepackage{xspace}

\usepackage{dsfont}

\usepackage{mathrsfs}

\usepackage{twoopt}

\providecommand*{\MathLib}{UNDEFINED}
\ifthenelse{\equal{\MathLib}{UNDEFINED}}{
\renewcommand*{\MathLib}{DEFINED}

\providecommand*{\Xmath}[1]{\ensuremath{#1}\xspace}

\providecommand*{\XmathThreePar}[3]{%
    \ifthenelse{\equal{#3}{'}}{%
      \Xmath{#1_{#2}'}
    }{%
      \Xmath{#1_{#2}^{#3}}
    }
}

\providecommand*{\XmathFourPar}[4]{\XmathThreePar{#1{#2}}{#3}{#4}}

\DeclareMathOperator{\diag}{diag}

\DeclareMathOperator{\rank}{rank}

\DeclareMathAccent{\Ring}{\mathord}{operators}{"17}

\newcommand*{\Tv}[3]{\XmathFourPar{}{#1}{#2}{#3}}

\newcommandtwoopt*{\Tvp}[2][][]{\Tv{p}{#1}{#2}}

\newcommandtwoopt*{\Tvz}[2][][]{\Tv{z}{#1}{#2}}

\newcommand*{\Constant}[3]{\XmathThreePar{#1}{#2}{#3}}

\newcommandtwoopt*{\ConstA}[2][][]{\Constant{a}{#1}{#2}}

\newcommandtwoopt*{\ConstB}[2][][]{\Constant{b}{#1}{#2}}

\newcommandtwoopt*{\ConstC}[2][][]{\Constant{c}{#1}{#2}}

\newcommandtwoopt*{\ConstD}[2][][]{\Constant{d}{#1}{#2}}

\newcommandtwoopt*{\ConstE}[2][][]{\Constant{e}{#1}{#2}}

\newcommandtwoopt*{\ConstK}[2][][]{\Constant{k}{#1}{#2}}

\newcommandtwoopt*{\ConstL}[2][][]{\Constant{l}{#1}{#2}}

\newcommandtwoopt*{\ConstM}[2][][]{\Constant{m}{#1}{#2}}

\newcommandtwoopt*{\ConstP}[2][][]{\Constant{p}{#1}{#2}}

\newcommandtwoopt*{\ConstQ}[2][][]{\Constant{q}{#1}{#2}}

\newcommandtwoopt*{\ConstU}[2][][]{\Constant{u}{#1}{#2}}

\newcommandtwoopt*{\ConstV}[2][][]{\Constant{v}{#1}{#2}}

\newcommandtwoopt*{\ConstW}[2][][]{\Constant{w}{#1}{#2}}

\newcommand*{\SpecialConstant}[3]{\XmathThreePar{\mathrm{#1}}{#2}{#3}}

\newcommandtwoopt*{\SConstE}[2][][]{\SpecialConstant{e}{#1}{#2}}

\newcommandtwoopt*{\Var}[2][][2]{\XmathThreePar{D}{#1}{#2}}

\newcommandtwoopt*{\Varx}[3][][2]{\XmathFourPar{#3}{D}{#1}{#2}}

\newcommandtwoopt*{\Cov}[2][][]{\XmathThreePar{R}{#1}{#2}}

\newcommandtwoopt*{\Covx}[3][][]{\XmathFourPar{#3}{R}{#1}{#2}}

\newcommandtwoopt*{\DenSx}[3][][]{\XmathFourPar{#3}{S}{#1}{#2}}

\newcommandtwoopt*{\DenS}[2][][]{\XmathThreePar{S}{#1}{#2}}

\newcommand*{\OdeSol}[3]{\XmathFourPar{}{#1}{#2}{#3}}

\newcommandtwoopt*{\OdeSolU}[2][][]{\OdeSol{u}{#1}{#2}}

\newcommandtwoopt*{\OdeSolV}[2][][]{\OdeSol{v}{#1}{#2}}

\newcommandtwoopt*{\OdeSolW}[2][][]{\OdeSol{w}{#1}{#2}}

\newcommandtwoopt*{\OdeSolX}[2][][]{\OdeSol{x}{#1}{#2}}

\newcommandtwoopt*{\OdeSolY}[2][][]{\OdeSol{y}{#1}{#2}}

\newcommandtwoopt*{\OdeSolZ}[2][][]{\OdeSol{z}{#1}{#2}}

\newcommand*{\NSet}[3]{\XmathFourPar{\mathbb}{#1}{#2}{#3}}

\newcommandtwoopt*{\A}[2][][]{\NSet{A}{#1}{#2}}

\newcommandtwoopt*{\B}[2][][]{\NSet{B}{#1}{#2}}

\newcommandtwoopt*{\C}[2][][]{\NSet{C}{#1}{#2}}

\newcommandtwoopt*{\F}[2][][]{\NSet{F}{#1}{#2}}

\newcommandtwoopt*{\K}[2][][]{\NSet{K}{#1}{#2}}

\newcommandtwoopt*{\HatK}[2][][]{\NSet{\Hat{K}}{#1}{#2}}

\newcommandtwoopt*{\N}[2][][]{\NSet{N}{#1}{#2}}

\newcommandtwoopt*{\Q}[2][][]{\NSet{Q}{#1}{#2}}

\newcommandtwoopt*{\R}[2][][]{\NSet{R}{#1}{#2}}

\newcommandtwoopt*{\BarR}[2][][]{\NSet{\overline{R}}{#1}{#2}}

\newcommandtwoopt*{\T}[2][][]{\NSet{T}{#1}{#2}}

\newcommandtwoopt*{\U}[2][][]{\NSet{U}{#1}{#2}}

\newcommandtwoopt*{\V}[2][][]{\NSet{V}{#1}{#2}}

\newcommandtwoopt*{\W}[2][][]{\NSet{W}{#1}{#2}}

\newcommandtwoopt*{\Z}[2][][]{\NSet{Z}{#1}{#2}}

\newcommand*{\Sa}[3]{\XmathFourPar{\mathfrak}{#1}{#2}{#3}}

\newcommandtwoopt*{\SaA}[2][][]{\Sa{A}{#1}{#2}}

\newcommandtwoopt*{\SaB}[2][][]{\Sa{B}{#1}{#2}}

\newcommandtwoopt*{\SaF}[2][][]{\Sa{F}{#1}{#2}}

\newcommandtwoopt*{\SaG}[2][][]{\Sa{G}{#1}{#2}}

\newcommandtwoopt*{\SaH}[2][][]{\Sa{H}{#1}{#2}}

\newcommandtwoopt*{\SaX}[2][][]{\Sa{X}{#1}{#2}}

\newcommandtwoopt*{\SaY}[2][][]{\Sa{Y}{#1}{#2}}

\newcommandtwoopt*{\SaZ}[2][][]{\Sa{Z}{#1}{#2}}

\newcommand*{\Rv}[4]{%
  \ifthenelse{\equal{#1}{EMPTY}}{%
    \XmathThreePar{#2}{#3}{#4}
  }{%
    \XmathFourPar{\skew4#1}{#2}{#3}{#4}
  }
}

\newcommandtwoopt*{\RvAx}[3][][]{\Rv{#3}{\alpha}{#1}{#2}}

\newcommandtwoopt*{\RvBx}[3][][]{\Rv{#3}{\beta}{#1}{#2}}

\newcommandtwoopt*{\RvCx}[3][][]{\Rv{#3}{\chi}{#1}{#2}}

\newcommandtwoopt*{\RvDx}[3][][]{\Rv{#3}{\delta}{#1}{#2}}

\newcommandtwoopt*{\RvEx}[3][][]{\Rv{#3}{\epsilon}{#1}{#2}}

\newcommandtwoopt*{\RvFx}[3][][]{\Rv{#3}{\varphi}{#1}{#2}}

\newcommandtwoopt*{\RvFVx}[3][][]{\Rv{#3}{\phi}{#1}{#2}}

\newcommandtwoopt*{\RvGx}[3][][]{\Rv{#3}{\gamma}{#1}{#2}}

\newcommandtwoopt*{\RvHx}[3][][]{\Rv{#3}{\theta}{#1}{#2}}

\newcommandtwoopt*{\RvIx}[3][][]{\Rv{#3}{\iota}{#1}{#2}}

\newcommandtwoopt*{\RvKx}[3][][]{\Rv{#3}{\kapa}{#1}{#2}}

\newcommandtwoopt*{\RvMx}[3][][]{\Rv{#3}{\mu}{#1}{#2}}

\newcommandtwoopt*{\RvNx}[3][][]{\Rv{#3}{\nu}{#1}{#2}}

\newcommandtwoopt*{\RvSx}[3][][]{\Rv{#3}{\varsigma}{#1}{#2}}

\newcommandtwoopt*{\RvTx}[3][][]{\Rv{#3}{\vartheta}{#1}{#2}}

\newcommandtwoopt*{\RvUx}[3][][]{\Rv{#3}{\psi}{#1}{#2}}

\newcommandtwoopt*{\RvVx}[3][][]{\Rv{#3}{\upsilon}{#1}{#2}}

\newcommandtwoopt*{\RvWx}[3][][]{\Rv{#3}{W}{#1}{#2}}

\newcommandtwoopt*{\RvXx}[3][][]{\Rv{#3}{\xi}{#1}{#2}}

\newcommandtwoopt*{\RvYx}[3][][]{\Rv{#3}{\eta}{#1}{#2}}

\newcommandtwoopt*{\RvZx}[3][][]{\Rv{#3}{\zeta}{#1}{#2}}

\newcommandtwoopt*{\RvA}[2][][]{\RvAx[#1][#2]{EMPTY}}

\newcommandtwoopt*{\RvB}[2][][]{\RvBx[#1][#2]{EMPTY}}

\newcommandtwoopt*{\RvC}[2][][]{\RvCx[#1][#2]{EMPTY}}

\newcommandtwoopt*{\RvD}[2][][]{\RvDx[#1][#2]{EMPTY}}

\newcommandtwoopt*{\RvE}[2][][]{\RvEx[#1][#2]{EMPTY}}

\newcommandtwoopt*{\RvF}[2][][]{\RvFx[#1][#2]{EMPTY}}

\newcommandtwoopt*{\RvFV}[2][][]{\RvFVx[#1][#2]{EMPTY}}

\newcommandtwoopt*{\RvG}[2][][]{\RvGx[#1][#2]{EMPTY}}

\newcommandtwoopt*{\RvH}[2][][]{\RvHx[#1][#2]{EMPTY}}

\newcommandtwoopt*{\RvI}[2][][]{\RvIx[#1][#2]{EMPTY}}

\newcommandtwoopt*{\RvK}[2][][]{\RvKx[#1][#2]{EMPTY}}

\newcommandtwoopt*{\RvM}[2][][]{\RvMx[#1][#2]{EMPTY}}

\newcommandtwoopt*{\RvN}[2][][]{\RvNx[#1][#2]{EMPTY}}

\newcommandtwoopt*{\RvS}[2][][]{\RvSx[#1][#2]{EMPTY}}

\newcommandtwoopt*{\RvT}[2][][]{\RvTx[#1][#2]{EMPTY}}

\newcommandtwoopt*{\RvU}[2][][]{\RvUx[#1][#2]{EMPTY}}

\newcommandtwoopt*{\RvV}[2][][]{\RvVx[#1][#2]{EMPTY}}

\newcommandtwoopt*{\RvW}[2][][]{\RvWx[#1][#2]{EMPTY}}

\newcommandtwoopt*{\RvX}[2][][]{\RvXx[#1][#2]{EMPTY}}

\newcommandtwoopt*{\RvY}[2][][]{\RvYx[#1][#2]{EMPTY}}

\newcommandtwoopt*{\RvZ}[2][][]{\RvZx[#1][#2]{EMPTY}}

\newcommand*{\Matrix}[4]{\XmathFourPar{#1}{\mathcal{#2}}{#3}{#4}}

\newcommandtwoopt*{\MtxAx}[3][][]{\Matrix{#3}{A}{#1}{#2}}

\newcommandtwoopt*{\MtxBx}[3][][]{\Matrix{#3}{B}{#1}{#2}}

\newcommandtwoopt*{\MtxCx}[3][][]{\Matrix{#3}{C}{#1}{#2}}

\newcommandtwoopt*{\MtxDx}[3][][]{\Matrix{#3}{D}{#1}{#2}}

\newcommandtwoopt*{\MtxEx}[3][][]{\Matrix{#3}{E}{#1}{#2}}

\newcommandtwoopt*{\MtxFx}[3][][]{\Matrix{#3}{F}{#1}{#2}}

\newcommandtwoopt*{\MtxGx}[3][][]{\Matrix{#3}{G}{#1}{#2}}

\newcommandtwoopt*{\MtxHx}[3][][]{\Matrix{#3}{H}{#1}{#2}}

\newcommandtwoopt*{\MtxIx}[3][][]{\Matrix{#3}{I}{#1}{#2}}

\newcommandtwoopt*{\MtxJx}[3][][]{\Matrix{#3}{J}{#1}{#2}}

\newcommandtwoopt*{\MtxLx}[3][][]{\Matrix{#3}{L}{#1}{#2}}

\newcommandtwoopt*{\MtxMx}[3][][]{\Matrix{#3}{M}{#1}{#2}}

\newcommandtwoopt*{\MtxNx}[3][][]{\Matrix{#3}{N}{#1}{#2}}

\newcommandtwoopt*{\MtxOx}[3][][]{\Matrix{#3}{O}{#1}{#2}}

\newcommandtwoopt*{\MtxPx}[3][][]{\Matrix{#3}{P}{#1}{#2}}

\newcommandtwoopt*{\MtxQx}[3][][]{\Matrix{#3}{Q}{#1}{#2}}

\newcommandtwoopt*{\MtxUx}[3][][]{\Matrix{#3}{U}{#1}{#2}}

\newcommandtwoopt*{\MtxVx}[3][][]{\Matrix{#3}{V}{#1}{#2}}

\newcommandtwoopt*{\MtxWx}[3][][]{\Matrix{#3}{W}{#1}{#2}}

\newcommandtwoopt*{\MtxXx}[3][][]{\Matrix{#3}{X}{#1}{#2}}

\newcommandtwoopt*{\MtxYx}[3][][]{\Matrix{#3}{Y}{#1}{#2}}

\newcommandtwoopt*{\MtxZx}[3][][]{\Matrix{#3}{Z}{#1}{#2}}

\newcommandtwoopt*{\MtxA}[2][][]{\MtxAx[{#1}][{#2}]{}}

\newcommandtwoopt*{\MtxB}[2][][]{\MtxBx[{#1}][{#2}]{}}

\newcommandtwoopt*{\MtxC}[2][][]{\MtxCx[{#1}][{#2}]{}}

\newcommandtwoopt*{\MtxD}[2][][]{\MtxDx[{#1}][{#2}]{}}

\newcommandtwoopt*{\MtxE}[2][][]{\MtxEx[{#1}][{#2}]{}}

\newcommandtwoopt*{\MtxF}[2][][]{\MtxFx[{#1}][{#2}]{}}

\newcommandtwoopt*{\MtxG}[2][][]{\MtxGx[{#1}][{#2}]{}}

\newcommandtwoopt*{\MtxH}[2][][]{\MtxHx[{#1}][{#2}]{}}

\newcommandtwoopt*{\MtxI}[2][][]{\MtxIx[{#1}][{#2}]{}}

\newcommandtwoopt*{\MtxJ}[2][][]{\MtxJx[{#1}][{#2}]{}}

\newcommandtwoopt*{\MtxL}[2][][]{\MtxLx[{#1}][{#2}]{}}

\newcommandtwoopt*{\MtxM}[2][][]{\MtxMx[{#1}][{#2}]{}}

\newcommandtwoopt*{\MtxN}[2][][]{\MtxNx[{#1}][{#2}]{}}

\newcommandtwoopt*{\MtxO}[2][][]{\MtxOx[{#1}][{#2}]{}}

\newcommandtwoopt*{\MtxP}[2][][]{\MtxPx[{#1}][{#2}]{}}

\newcommandtwoopt*{\MtxQ}[2][][]{\MtxQx[{#1}][{#2}]{}}

\newcommandtwoopt*{\MtxU}[2][][]{\MtxUx[{#1}][{#2}]{}}

\newcommandtwoopt*{\MtxV}[2][][]{\MtxVx[{#1}][{#2}]{}}

\newcommandtwoopt*{\MtxW}[2][][]{\MtxWx[{#1}][{#2}]{}}

\newcommandtwoopt*{\MtxX}[2][][]{\MtxXx[{#1}][{#2}]{}}

\newcommandtwoopt*{\MtxY}[2][][]{\MtxYx[{#1}][{#2}]{}}

\newcommandtwoopt*{\MtxZ}[2][][]{\MtxZx[{#1}][{#2}]{}}

\newcommand*{\Vector}[3]{\XmathThreePar{#1}{#2}{#3}}

\newcommandtwoopt*{\VecA}[2][][]{\Vector{a}{#1}{#2}}

\newcommandtwoopt*{\VecB}[2][][]{\Vector{b}{#1}{#2}}

\newcommandtwoopt*{\VecC}[2][][]{\Vector{c}{#1}{#2}}

\newcommandtwoopt*{\VecD}[2][][]{\Vector{d}{#1}{#2}}

\newcommandtwoopt*{\VecE}[2][][]{\Vector{e}{#1}{#2}}

\newcommandtwoopt*{\VecU}[2][][]{\Vector{u}{#1}{#2}}

\newcommandtwoopt*{\VecV}[2][][]{\Vector{v}{#1}{#2}}

\newcommandtwoopt*{\VecW}[2][][]{\Vector{w}{#1}{#2}}

\newcommandtwoopt*{\VecX}[2][][]{\Vector{x}{#1}{#2}}

\newcommandtwoopt*{\VecY}[2][][]{\Vector{y}{#1}{#2}}

\newcommandtwoopt*{\VecZ}[2][][]{\Vector{z}{#1}{#2}}

\newcommand*{\Sf}[4]{\XmathFourPar{#1}{#2}{#3}{#4}}

\newcommandtwoopt*{\Sfcx}[3][][]{\Sf{#3}{\varphi}{#1}{#2}}

\newcommandtwoopt*{\Sfgx}[3][][]{\Sf{#3}{g}{#1}{#2}}

\newcommandtwoopt*{\SfGx}[3][][]{\MtxGx[{#1}][{#2}]{#3}}

\newcommandtwoopt*{\Sffx}[3][][]{\Sf{#3}{\phi}{#1}{#2}}

\newcommandtwoopt*{\SfFx}[3][][]{\MtxFx[{#1}][{#2}]{#3}}

\newcommandtwoopt*{\Sfhx}[3][][]{\Sf{#3}{h}{#1}{#2}}

\newcommandtwoopt*{\SfHx}[3][][]{\MtxHx[{#1}][{#2}]{#3}}

\newcommandtwoopt*{\Sfc}[2][][]{\Sfcx[{#1}][{#2}]{}}

\newcommandtwoopt*{\Sfg}[2][][]{\Sfgx[{#1}][{#2}]{}}

\newcommandtwoopt*{\SfG}[2][][]{\SfGx[{#1}][{#2}]{}}

\newcommandtwoopt*{\Sff}[2][][]{\Sffx[{#1}][{#2}]{}}

\newcommandtwoopt*{\SfF}[2][][]{\SfFx[{#1}][{#2}]{}}

\newcommandtwoopt*{\Sfh}[2][][]{\Sfhx[{#1}][{#2}]{}}

\newcommandtwoopt*{\SfH}[2][][]{\SfHx[{#1}][{#2}]{}}

\newcommand*{\FSet}[3]{\XmathThreePar{#1}{#2}{#3}}

\newcommandtwoopt*{\FSetB}[2][][]{\FSet{B}{#1}{#2}}

\newcommandtwoopt*{\FSetBV}[2][][]{\FSet{BV}{#1}{#2}}

\newcommandtwoopt*{\FSetC}[2][][]{\FSet{C}{#1}{#2}}

\newcommandtwoopt*{\FSetF}[2][][]{\FSet{F}{#1}{#2}}

\newcommandtwoopt*{\FSetH}[2][][]{\FSet{H}{#1}{#2}}

\newcommandtwoopt*{\FSetI}[2][][]{\FSet{I}{#1}{#2}}

\newcommandtwoopt*{\FSetl}[2][][]{\FSet{l}{#1}{#2}}

\newcommandtwoopt*{\FSetL}[2][][]{\FSet{L}{#1}{#2}}

\newcommandtwoopt*{\FSetM}[2][][]{\FSet{M}{#1}{#2}}

\newcommandtwoopt*{\FSetR}[2][][]{\FSet{R}{#1}{#2}}

\newcommandtwoopt*{\FSetX}[2][][]{\FSet{X}{#1}{#2}}

\newcommandtwoopt*{\FSetY}[2][][]{\FSet{Y}{#1}{#2}}

\newcommand*{\OpSet}[3]{\XmathThreePar{\mathfrak{#1}}{#2}{#3}}

\newcommandtwoopt*{\OpSetS}[2][][]{\OpSet{S}{#1}{#2}}

\newcommand*{\XBrace}[4][]{\Xmath{\mathopen#1#2#4\mathclose#1#3}}

\newcommand*{\Norm}[2][]{\XBrace[#1]{\lVert}{\rVert}{#2}}

\newcommandtwoopt*{\MtxNorm}[3][][M]{\Norm[#1]{#3}_{#2}}

\newcommandtwoopt*{\SupNorm}[3][][S]{\Norm[#1]{#3}_{#2}}

}{} %** \ifthenelse at the beginning of this file (don't remove the brackets) **

\providecommand*{\ReportMathLib}{UNDEFINED}
\ifthenelse{\equal{\ReportMathLib}{UNDEFINED}}{
\renewcommand*{\ReportMathLib}{DEFINED}

\newcommand*{\transpose}[1]{\Xmath{#1'}}

\newcommand*{\abs}[1]{\Xmath{\left|#1\right|}}

\newcommand*{\norm}[1]{\Xmath{\|#1\|}}

\newcommand*{\idmat}[1]{\Xmath{I_{#1}}}

\newcommand*{\vol}[1]{\Xmath{\operatorname{Vol}\left(#1\right)}}

\DeclareMathOperator*{\argmin}{arg\;min}

\newcommand*{\setRealNumbers}{\Xmath{\mathbb{R}}}
\newcommand*{\setIntegers}{\Xmath{\mathbb{Z}}}

\newcommand*{\Field}{\Xmath{\mathbb{F}}}

}{} %* \ifthenelse at the beginning of this file (don't remove the brackets) ***

\newacronym{MRT}{MRT}{maximum ratio transmission}
\newacronym{AWGN}{AWGN}{additive white Gaussian noise}
\newacronym{SISO}{SISO}{single-input single-output}
\newacronym{MISO}{MISO}{multiple-input single-output}
\newacronym{SIMO}{SIMO}{single-input multiple-output}
\newacronym{MIMO}{MIMO}{multiple-input multiple-output}
\newacronym{SNR}{SNR}{signal-to-noise ratio}
\newacronym{MAC}{MAC}{multiple access channel}

\begin{document}

% remove page number on all pages (except title page)
\pagestyle{empty}

% remove page number on title page
\makeatletter
\def\ps@IEEEtitlepagestyle{%
\def\@oddhead{}%
\def\@evenhead{}%
\def\@oddfoot{}%
\def\@evenfoot{}}
\makeatother

% stretch baseline to get 28 lines of text per page
\renewcommand{\baselinestretch}{1.75}
\normalsize

\sloppy

%% Paper Title
\title{Weak Secrecy in the Multi-Way Untrusted Relay Channel with Compute-and-Forward}

%% Authors
\author{
  Johannes~Richter,~Christian Scheunert,~Sabrina Engelmann, and Eduard~A.~Jorswieck
  \thanks{This work is supported by the German Research Foundation (DFG) in the Collaborative Research Center 912 ``Highly Adaptive Energy-Efficient Computing'' and within the Cluster of Excellence ``Center for Advancing Electronics Dresden''.\par The authors are with the Department of Electrical Engineering and Information Technology, Technische Universität Dresden, 01062 Dresden, Germany. Email: \{johannes.richter, sabrina.engelmann, christian.scheunert, eduard.jorswieck\}@tu-dresden.de}
}
% \author{
%   \IEEEauthorblockN{Johannes Richter, Christian Scheunert, Sabrina Engelmann, and Eduard A. Jorswieck}
%   \IEEEauthorblockA{Dep. of Electrical Engineering and Information Technology / Communications Laboratory\\
%     Technische Universität Dresden, 01062 Dresden, Germany\\
%     Email: \{Johannes.Richter, Christian.Scheunert, Sabrina.Engelmann, Eduard.Jorswieck\}@tu-dresden.de}
%   \thanks{This work is partly supported by the German Research Foundation (DFG) in the Collaborative Research Center 912 ``Highly Adaptive Energy-Efficient Computing'' and within the Cluster of Excellence ``Center for Advancing Electronics Dresden''.}
% }

%% Create Title
\maketitle

%% Abstract
\begin{abstract}
We investigate the problem of secure communications in a Gaussian multi-way relay channel applying the compute-and-forward scheme using nested lattice codes.
All nodes employ half-duplex operation and can exchange confidential messages only via an untrusted relay.
The relay is assumed to be honest but curious, i.e., an eavesdropper that conforms to the system rules and applies the intended relaying scheme.

We start with the general case of the single-input multiple-output (SIMO) L-user multi-way relay channel and provide an achievable secrecy rate region under a weak secrecy criterion.
We show that the securely achievable sum rate is equivalent to the difference between the computation rate and the multiple access channel (MAC) capacity.
Particularly, we show that all nodes must encode their messages such that the common computation rate tuple falls outside the MAC capacity region of the relay.
We provide results for the single-input single-output (SISO) and the multiple-input single-input (MISO) L-user multi-way relay channel as well as the two-way relay channel.
We discuss these results and show the dependency between channel realization and achievable secrecy rate.
We further compare our result to available results in the literature for different schemes and show that the proposed scheme operates close to the compute-and-forward rate without secrecy.
\end{abstract}


\section{Introduction}
\label{sec:introduction}

Network coding has been a promising topic in communications since introduced by Ahlswede et al.\ in \cite{Ahlswede2000}.
It was shown that network coding can improve the throughput of a network and achieves the multicast capacity.
For static wired networks network coding is very well investigated and some frameworks are developed \cite{Fragouli2007,Yeung2008,Yeung2005,Yeung2005a}.
On the other hand, there is a lot of ongoing work in the field of network coding for wireless networks.
The properties of the wireless channel give the possibility for network coding on different layers.
Practical network coding on the forwarding layer has been proposed in \cite{Katti2008}.
The superposition property of the wireless channel also allows network coding on the physical layer, where the actual network coding is already done by the channel.
Physical layer network coding has been investigated in \cite{Zhang2006} and gained a lot of attention.
\cite{Nazer2011a} gives a survey on physical layer network coding techniques.
In \cite{Nazer2011}, this approach has been further developed to the compute-and-forward scheme for which the noise is immediately removed at any relay node in the network.
This is achieved by decoding linear combinations of incoming symbols at a relay instead of decoding them individually.
Compute-and-forward is based on structured codes like lattice codes that have been shown to achieve the additive white Gaussian noise channel capacity with lattice decoding in place of maximum-likelihood decoding in \cite{Erez2004}.

Beside reliable communication, secrecy considerations have also become more important.
Cryptography is based on the currently available computation performance and the time needed to decrypt a message without having the secret key.
This kind of security will get weaker with increasing computational power in the years to come.
Secrecy on the physical layer offers the possibility that an eavesdropper does not get any information about the exchanged messages.
Wyner \cite{Wyner1975} and Csisz\'{a}r and Körner \cite{Csiszar1978} proved that confidential data transmission over wiretap channels can be attained by channel coding without secret keys in terms of weak and strong secrecy, respectively.
% In cryptography, transmission channels are usually considered to be free of physical layer effects and hence they are assumed to be noiseless.
% In physical-layer security, noise present on the channel can be used to enhance the level of security of a communication system.
The achievable secrecy rate in the presence of an untrusted relay is studied in several papers with slightly different scenarios.
In \cite{Huang2013} Huang et al.\ investigate different secure transmission schemes in a relay network with a direct connection between source and destination.
The question if an untrusted relay is helpful if a direct connection between source and destination exists has been investigated in \cite{He2010}.
In \cite{He2013a} a Gaussian two-hop network is considered where source and destination do not have a direct connection.
The destination node can help the source node by jamming the relay node with a random signal.
This model is extended in \cite{He2013} to a multi-hop line network.

The two-way wiretap channel, in which two nodes can only exchange messages via an untrusted relay, was first considered in \cite{Tekin2008,Tekin2010}.
It was shown that cooperative jamming, i.e., jamming with controlled interference between codewords, could reduce the eavesdropper's signal-to-noise ratio and hence improve the level of weak secrecy.
In \cite{Pierrot2011} it was shown that this result also holds for strong secrecy.

Most secrecy schemes are based on random codes but there is also work done in the field of secrecy through structured codes, namely lattice codes. 
Achievable rates for the lattice coded Gaussian wiretap channel have been developed in \cite{Choo2011}.
Theorem 1 in \cite{Choo2011} proposes a lattice code construction that achieves the weak secrecy capacity.
Lattice codes for the Gaussian wiretap channel have been intensively studied by Oggier et al.\ in \cite{Oggier2013} and references therein.
They introduced the secrecy gain as a design criterion for good lattice codes for wiretap channels in \cite{Belfiore2010}.
It was shown by Ling et al.\ that lattice codes can achieve strong secrecy over the mod-$\Lambda$ Gaussian channel \cite{Ling2012}.
In \cite{Ling2012a} Ling et al. introduced the flatness factor \cite{Belfiore2011} as the main tool to prove that nested lattice codes can achieve semantic security and strong secrecy over the Gaussian wiretap channel.
Compute-and-forward network coding together with strong physical-layer security based on universal hash functions has been investigated in \cite{He2013a}.
All these works focus either on a wiretap channel or on a two-hop relay network where a source node transmits via an untrusted relay to a destination.
The destination may help by jamming but does not transmit a secure message itself.
In \cite{Kashyap2012} Kashyap et al.\ consider a two-way relay network with an untrusted relay where two nodes transmit one message each simultaneously via an untrusted relay.
They provide an achievable power-rate region with perfect secrecy as well as strong secrecy by applying compute-and-forward.

In this work we consider a L-user relay channel where all users want to securely transmit a message to all other users.
There does not exist direct connections between the users and they have to transmit via an untrusted relay which applies compute-and-forward.
We investigate the \gls{SIMO} multi-way relay channel and provide an achievable weak secrecy rate region.
We prove this result and derive results for the \gls{SISO} and \gls{MISO} multi-way relay channel as well as the \gls{SISO} two-way relay channel, which can be seen as special cases of the \gls{SIMO} multi-way relay channel.
Finally, we provide simulation results and compare our scheme against existing schemes.

Our scheme has several advantages over existing schemes.
The previous work of \cite{He2008} and \cite{Kashyap2012} does not consider fading channels and their results can not be easily extended.
Our scheme however takes fading into account and uses only well-known techniques like random binning for wiretap channels and compute-and-forward for two-way relaying.
Further, we provide an descriptive and convincing interpretation of our result, i.e., the secrecy rate is the compute-and-forward rate region minus the \gls{MAC} rate region.
We also achieve a slightly higher secrecy rate when we simplify our model to the model used for example in \cite{He2008}.
A drawback of our scheme is the weak secrecy criterion instead of strong or even perfect secrecy.
However, this disadvantage might be overcome by extending our scheme with the same techniques as in \cite{He2013a}, namely hashing functions, to provide strong secrecy.
This will be part of future work and is not addressed in this paper.
 
% We investigate the two-way relay channel as a very important special case of this scenario.
% We provide an achievable weak secrecy rate region and compare the results to existing schemes.
% Finally, we extend the results to the general L-user case and provide a full proof of the achievable secrecy rate region. 
% In this work we  consider the Gaussian two-way relay channel applying the compute-and-forward scheme under weak secrecy criterion.
% We show that the securely achievable sum rate is equivalent to the difference between computation rate and MAC capacity.
% We provide an achievable secrecy rate region including the proof and compare this result to \cite{Kashyap2012}.

The outline of the paper is as follows. 
In Section \ref{sec:system-model} the system model and the coding scheme are given. 
Section \ref{sec:simo-l-user} provides an achievable secrecy rate region for the L-user \gls{SIMO} multi-way relay channel.
Section \ref{sec:siso-miso-l-user} derives from the previous result the achievable secrecy rate region for the \gls{SISO} and the \gls{MISO} multi-way relay channel.
Section \ref{sec:2-user-case} investigates the special case of the \gls{SISO} two-way relay channel.
Section \ref{sec:discussion} discusses the results and illustrates them with simulations.
Section \ref{sec:conclusion} concludes the work.

\subsection{Notation}

Let $\log^+(x) \triangleq \max\{0,\log(x)\}$.
We denote by $\transpose{x}$ the transpose of vector $x$ and by $e_i$ the unit vector with a one at position $i$ and zeros elsewhere.
Further we denote by $H(X)$ the entropy of a discrete random variable $X$ and by $h(Y)$ the differential entropy of a continuous random variable $Y$.
The mutual information between two random variables $X$ and $Y$ is denoted by $I(X;Y)$.
Let $\mathcal{P}(X)$ be the power set of set $X$.
By $\delta(n)$ we denote a function that tends to zero if $n$ goes to infinity.

% In the following we recall some lattice definitions that are used throughout the paper. 
% An $n$-dimensional lattice $\Lambda \subset \setRealNumbers^n$ is a group under addition with generator matrix $G \in \setRealNumbers^{n \times n}$:
% \begin{equation}
%   \Lambda = \{Gc : c \in \setIntegers^n\}.
% \end{equation}
% A lattice quantizer is a mapping $Q_\Lambda: \setRealNumbers^n \to \Lambda$ that maps a point $x$ to the nearest lattice point in Euclidean distance,
% \begin{equation}
%   Q_\Lambda(x) = \argmin_{\lambda \in \Lambda} \norm{x - \lambda}. 
% \end{equation}
% Let the modulo operation with respect to the lattice $\Lambda$ be defined as 
% \begin{equation}
%   x \bmod \Lambda = x - Q_\Lambda(x).
% \end{equation}
% We call $\mathcal{V} = \{x:Q_\Lambda(x) = 0\}$ the fundamental Voronoi region of the lattice $\Lambda$ and denote by $\vol{\mathcal{V}}$ the volume of $\mathcal{V}$.

\section{System Model}
\label{sec:system-model}

\begin{figure}
  \centering
  \tikzsetnextfilename{l-user-relay-channel}
\begin{tikzpicture}[>=latex]

  \tikzpicturedependsonfile{figures/l-user-relay-channel.tikz}
  \tikzset{
    nn/.style={
      draw,
      circle,
      inner sep=0,
      minimum size=2em,
    }
  }

  % MAC Phase
  \begin{scope}
    \node at (0,1.5) {1. Phase: MAC};
    \node[nn](n1) at (0,0) {$R$};
    \node[nn](n2) at (150:1.5) {$1$};
    \node[nn](n3) at (30:1.5) {$2$};
    \node[nn](n4) at (-90:1.5) {$L$};
    \node[rotate=-30](n5) at (-30:1) {$\vdots$};

    \draw[->] (n2) -- node[above] {$h_1$} (n1);
    \draw[->] (n3) -- node[above] {$h_2$} (n1);
    \draw[->] (n4) -- node[left] {$h_L$} (n1);
  \end{scope}

  % BC Phases
  \begin{scope}[xshift=4cm]
    \node at (0,1.5) {2.-L. Phase: BC};
    \node[nn](n1) at (0,0) {$R$};
    \node[nn](n2) at (150:1.5) {$1$};
    \node[nn](n3) at (30:1.5) {$2$};
    \node[nn](n4) at (-90:1.5) {$L$};
    \node[rotate=-30](n5) at (-30:1) {$\vdots$};

    \draw[<-] (n2) -- node[above] {$h_1$} (n1);
    \draw[<-] (n3) -- node[above] {$h_2$} (n1);
    \draw[<-] (n4) -- node[left] {$h_L$} (n1);
  \end{scope}
  
\end{tikzpicture}
  \caption{System model: Multi-Way Relay-Channel}
  \label{fig:system-model}
\end{figure}
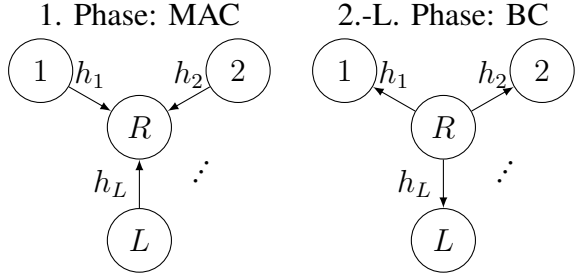
We investigate the multi-way relay-channel with half-duplex nodes as depicted in Figure \ref{fig:system-model}.
All nodes have messages for each other but have no direct connection.
They transmit messages $w_1,\dots,w_L$ with the help of a relay in $L$ phases.
We denote by $x_\ell \in \setRealNumbers^n$ and $x_r \in \setRealNumbers^n$ the transmit signals of the user nodes $1,\dots,L$ and the relay, respectively.
Analogously, we define $y_\ell \in \setRealNumbers^n$ and $y_r \in \setRealNumbers^n$ as the received signals at the user nodes and the relay. 
Each node has a transmit power constraint $\norm{x_i}^2 \leq nP,~\forall i \in \{1,\dots,L,r\}$. 
The links are \gls{AWGN} channels with quasi static block flat fading with fading coefficients $h_i$ and noise $z_i$ with $i \in \{1,\dots,L,r\}$.
We assume that the channels are reciprocal and constant over all $L$ phases.
This results in the following channel model
\begin{subequations}
  \begin{align}
    y_r &= \sum_{\ell = 1}^L h_\ell x_\ell + z_r,\\
    y_\ell &= h_\ell x_r + z_\ell.
  \end{align}
\end{subequations}
We consider four scenarios:
\begin{enumerate}
\item multiple antennas at the relay, which is covered in Section \ref{sec:simo-l-user},
\item multiple antennas at the user nodes, which is covered in Section \ref{sec:siso-miso-l-user},
\item single antennas at all nodes, which is covered in Section \ref{sec:siso-miso-l-user}, and
\item the two-way relay channel with single antennas at all nodes as a special but important case, which is covered in Section \ref{sec:2-user-case}.
\end{enumerate}
In practice it often happens that such a transmission between nodes has to use a relay that cannot be trusted.
Therefore the messages have to be encoded at the source nodes such that the relay cannot decode the messages separately.
% In this work we provide a transmission strategy based on compute-and-forward where all source nodes want to securely transmit messages over a relay where the relay is an honest but curious eavesdropper.
The reliability requirement at the user nodes $1,\dots,L$ can be written as
\begin{equation}
   \lim_{n \to \infty} \Pr(\hat{w}_\ell \neq w_\ell) = 0,~\forall \ell = 1,\dots,L
\end{equation}
and the weak secrecy requirement as
\begin{equation}
  \lim_{n \to \infty} \tfrac{1}{n} I(\mathcal{W} ; Y_r) = 0,
  \label{eq:weak-secrecy-criterion}
\end{equation} 
for all $\mathcal{W} \in \mathcal{P}(\{W_1,\dots,W_L\})$ where $n$ is the block length or number of channel uses.
% $I(W_1,\dots,W_L;Y)$ is the mutual information between the messages $W_1,\dots,W_L$ of the users and the received signal $Y_r$ at the eavesdropper and $n$ is the block length or number of channel uses.
For the analysis of the secrecy condition we will also use the following alternative representation of \eqref{eq:weak-secrecy-criterion}
\begin{equation}
  \label{eq:weak-secrecy-criterion-alternative}
  \lim_{n \to \infty} \tfrac{1}{n} H(\mathcal{W}) \leq \lim_{n \to \infty} H(\mathcal{W} \mid Y).
\end{equation}
We get the achievable secrecy rates $R_{s_1},\dots,R_{s_L}$ by choosing $H(W_\ell) = 2^{nR_{s_\ell}}$ such that
\begin{equation}
  \lim_{n \to \infty} \tfrac{1}{n} \log_2 H(W_\ell) = R_{s_\ell}.
\end{equation}

% Throughout this paper we will consider weak secrecy as defined by the following equation
% \begin{equation}
%   \label{eq:weak-secrecy-criterion}
%   \lim_{n \to \infty} \tfrac{1}{n} I(W ; Y) = 0,
% \end{equation}
% where $I(W;Y)$ is the mutual information between a message $W$ of a transmitter and the received signal $Y$ at the eavesdropper and $n$ is the block length or number of channel uses.
% For the analysis of the secrecy condition we will also use the following alternative representation of \eqref{eq:weak-secrecy-criterion}
% \begin{equation}
%   \label{eq:weak-secrecy-criterion-alternative}
%   \lim_{n \to \infty} \tfrac{1}{n} H(W) \leq \lim_{n \to \infty} H(W \mid Y).
% \end{equation}
Because the relay is the indented receiver of the messages and the eavesdropper at the same time, some additional effort is needed to ensure secret message transmission.
We can already see, that we can achieve this condition only if we apply a relaying scheme where the relay does not need to decode the single messages.
Further we need to achieve a transmission rate larger than the \gls{MAC} capacity otherwise the relay would be able to decode the single messages.
This is a necessary but not a sufficient condition.
A relaying strategy which allows to fulfill these requirements is compute-and-forward, which was introduced by Nazer and Gastpar in \cite{Nazer2011}.
We will use this framework to provide an achievable secrecy rate region.   

\subsection{Nested Lattice Code}
\label{sec:nested-lattice-code}

% We base our transmission scheme on nested lattice codes because they have a linear structure and achieve the capacity of an \gls{AWGN} channel \cite{Erez2004}.
The compute-and-forward framework is based on nested lattice codes and therefore we recall some lattice definitions that are used throughout the paper.
For further details on lattice codes see \cite{Conway1999,Zamir2009,Nazer2011a,Erez2005a}.

An $n$-dimensional lattice $\Lambda \subset \setRealNumbers^n$ is a group under addition with generator matrix $G \in \setRealNumbers^{n \times n}$:
\begin{equation}
  \Lambda = \{Gc : c \in \setIntegers^n\}.
\end{equation}
A lattice quantizer is a mapping $Q_\Lambda: \setRealNumbers^n \to \Lambda$ that maps a point $x$ to the nearest lattice point in Euclidean distance,
\begin{equation}
  Q_\Lambda(x) = \argmin_{\lambda \in \Lambda} \norm{x - \lambda}. 
\end{equation}
Let the modulo operation with respect to the lattice $\Lambda$ be defined as 
\begin{equation}
  x \bmod \Lambda = x - Q_\Lambda(x).
\end{equation}
We call $\mathcal{V} = \{x:Q_\Lambda(x) = 0\}$ the fundamental Voronoi region of the lattice $\Lambda$ and denote by $\vol{\mathcal{V}}$ the volume of $\mathcal{V}$.
Two lattices $\Lambda_C$ and $\Lambda_F$ are called nested, if $\Lambda_C \subseteq \Lambda_F$.
We call $\Lambda_C$ the coarse lattice and $\Lambda_F$ the fine lattice.
A nested lattice code $\mathcal{L}$ is formed by taking all of the points of the fine lattice $\Lambda_F$ in the fundamental Voronoi region $\mathcal{V}_C$ of the coarse lattice $\Lambda_C$, i.e., $\mathcal{L} = \Lambda_F \cap \mathcal{V}_C$.
The rate of a nested lattice code is given by
\begin{equation}
  r = \frac{1}{n}\log \abs{\mathcal{L}} = \frac{1}{n} \log \frac{\vol{\mathcal{V}_C}}{\vol{\mathcal{V}_F}}.
\end{equation}
From \cite{Erez2004} and \cite{Erez2005a} we know that there exist good nested lattice codes that can achieve the capacity of an \gls{AWGN} channel.
These nested lattice codes have been used to develop the compute-and-forward framework \cite{Nazer2011,Nazer2011a}. 
We utilize this framework as relaying strategy and extend it with secrecy constraints as described above.
%
% In the multi-way relay channel, we consider compute-and-forward at the relay as introduced in \cite{Nazer2011}.
% Thereby, each message is encoded at the user nodes using a nested lattice code $\mathcal{L}$.
% For further details on lattice codes see \cite{Zamir2009,Nazer2011a,Erez2005a}.

\subsection{Encoding}

Each user node $\ell$ chooses a message $w_\ell \in \Field_p^k$ i.i.d. from a uniform distribution over the index set $\{1,2,\dots,2^{\lfloor nR_s \rfloor } \}$.
For the ease of simplicity we assume equal message length at all users.
If this is no the case, messages with length smaller than $k$ will be padded to length $k$ with zeros.
Each message gets mapped to a lattice code $\mathcal{L} = \Lambda_F \cap \mathcal{V}_C$ where the second moment of $\Lambda_C$ equals $P$ such that the power constraint is satisfied.
In order to fulfill the secrecy requirements some additional effort is required. 
Each user node $\ell$ uses the same codebook $\mathcal{L} = \Lambda_F \cap \mathcal{V}_C$ with $|\Lambda_F \cap \mathcal{V}_C| = 2^{\lfloor n(R_s + R_d) \rfloor}$.
Like wiretap codes, this codebook is randomly binned into several bins, where each bin contains $2^{\lfloor nR_d \rfloor}$ codewords.
The secret message $w_\ell$ gets mapped to the bins.
The actual transmitted codeword $t_\ell$ is chosen from that bin according to a uniform distribution.
%
% Then, any message gets mapped to a lattice point $t_\ell = \phi(w_\ell) \in \Lambda_F \cap \mathcal{V}_B$ by an encoding function $\phi$.
% The encoding function $\phi$ can be denoted by \cite[Lemma 5]{Nazer2011}
% \begin{equation}
%   \label{eq:message-to-lattice}
%   \phi(x) = [G_B p^{-1} g(G x)] \bmod \Lambda_B, \,x \in \Field_p^k
% \end{equation}
% where $G_B \in \setRealNumbers^{n \times n}$ is the generator matrix of $\Lambda_B$, $G \in \Field_p^{n \times k}$ is a code generator matrix, and $g$ is a component-wise mapping
% \begin{equation}
%   g: \Field_p \to \{0,1,2,\dots,p-1\}.
% \end{equation}
% Basically, $g$ provides a change of alphabet by keeping the set of elements.
% %
% If $p$ grows appropriately with $n$ such that $n/p \to 0$ as $n \to \infty$ and we use the mapping in \eqref{eq:message-to-lattice} we can achieve the computation rate as given in \cite[Theorem 2]{Nazer2011}.
% A more general encoding and mapping strategy for finite $n$ and practical physical-layer network coding schemes is given in \cite[Section V]{Feng2012}.

% We add an additional lattice point $v_\ell$ that is chosen i.i.d. from a uniform distribution over $\Lambda_B \cap \mathcal{V}_C$ to the lattice point $t_\ell$.
% This lattice point is introduced to confuse the eavesdropper such that it is not able to get information about the sent messages.
% Adding $v_\ell$ is equivalent to randomly choosing a coset of $\mathcal{L}$.
%
Further we add some dither $u_\ell$ that is uniform distributed over $\mathcal{V}_C$ and known by the relay.
This dithering gives statistical properties of the transmitted signal needed to achieve the compute-and-forward rate \cite{Nazer2011}.
It was shown in \cite[Appendix C]{Nazer2011} that this dither might be chosen in a deterministic way.
To make sure that the transmit signal fulfills the power constraint, we build the modulo with respect to the coarse lattice.
We get the following $n$-dimensional transmit vector at node $\ell$
\begin{equation}
  x_\ell = [t_\ell + u_\ell] \bmod \Lambda_C .
\end{equation}
With this encoding scheme we get the following rates:
\begin{enumerate}[a)]
\item $R_d$ is the rate of the randomly chosen messages within a bin, 
\item $R_s$ is the secret message rate, and
\item $R_s + R_d = \frac{1}{n} \log_2 \frac{\vol{\mathcal{V}_C}}{\vol{\mathcal{V}_F}}$ is the transmit rate of the user nodes.
\end{enumerate}

\section{SIMO Multi-Way Relay Channel}
\label{sec:simo-l-user}

In the following section we investigate the \gls{SIMO} channel, i.e., the relay is equipped with multiple antennas.
The system model for the first phase from the user nodes to the relay is shown in Figure \ref{fig:system-model-simo-case}.
We want to point out, that this is the most general case without assuming \gls{MIMO} channels because the \gls{SISO} and \gls{MISO} channel can be seen as special cases as described later.

\begin{figure*}
  \centering
  \tikzsetnextfilename{system-model-simo-l-user-relay-channel}
\begin{tikzpicture}[>=latex,node distance=4mm]

  \tikzpicturedependsonfile{figures/system-model-simo-l-user-relay-channel.tikz}

  \tikzset{
    fitnode/.style={
      draw,
      dashed,
      rectangle,
      rounded corners,
      inner sep=2mm,
    },
    operator/.style={
      circle,
      draw,
      inner sep=0mm
    },
    function/.style={
      rectangle,
      draw
    },
  }

  \matrix (system) [
  matrix of math nodes,
  column sep=10mm,
  row sep=12mm,
  nodes in empty cells = true,
  nodes={% General options for all nodes
    anchor=center,
    text centered,
  },
  ] {
    w_1 & |[function]| \mathcal{E}_1 & |[minimum width=1cm]| & |[operator]| + & |[minimum width=1cm]| & |[function]| \mathcal{D}_1 & \hat{y}_1 \\
    w_L & |[function]| \mathcal{E}_L & |[minimum width=1cm]| & |[operator]| + & |[minimum width=1cm]| & |[function]| \mathcal{D}_{L-1} & \hat{y}_{L-1} \\
  };

  \node[draw,inner sep=0mm,fit=(system-1-3)(system-2-3),function]{$H$};
  \node[draw,inner sep=0mm,fit=(system-1-5)(system-2-5),function]{$B$};

  % Chains
  { [start chain,every on chain/.style={join}, every join/.style={->}]
    \foreach \y in {1,...,7}{
      \chainin (system-1-\y);
    }
  }

  { [start chain,every on chain/.style={join}, every join/.style={->}]
    \foreach \y in {1,...,7}{
      \chainin (system-2-\y);
    }
  }

  % additional noise
  \node(n1)[above=of system-1-4] {$z_1$};
  \draw[->] (n1) -- (system-1-4);
  \node(nL)[above=of system-2-4] {$z_{\eta_R}$};
  \draw[->] (nL) -- (system-2-4);

  % annotate paths
  \path[draw=none] (system-1-2) -- node[above] {$x_1$} (system-1-3);
  \path[draw=none] (system-2-2) -- node[above] {$x_L$} (system-2-3);

  \path[draw=none] (system-1-4) -- node[above] {$y_1$} (system-1-5);
  \path[draw=none] (system-2-4) -- node[above] {$y_{\eta_R}$} (system-2-5);

  \path[draw=none] (system-1-5) -- node[above] {$\tilde{y}_1$} (system-1-6);
  \path[draw=none] (system-2-5) -- node[above] {$\tilde{y}_{L-1}$} (system-2-6);

  % mark system parts
  \node[fitnode,label=above:User nodes,fit={(system-1-1)(system-2-2)(-6,1.7)(-6,-1.4)}] {};
  \node[fitnode,label=above:Channel,fit={(system-1-3)(system-2-4)(-1,1.7)(-1,-1.4)}] {};
  \node[fitnode,label=above:Relay,fit={(system-1-5)(system-2-7)(6,1.7)(6,-1.4)}] {};

\end{tikzpicture}
  \caption{System model of a SIMO L-user relay channel (MAC phase)}
  \label{fig:system-model-simo-case}
\end{figure*}
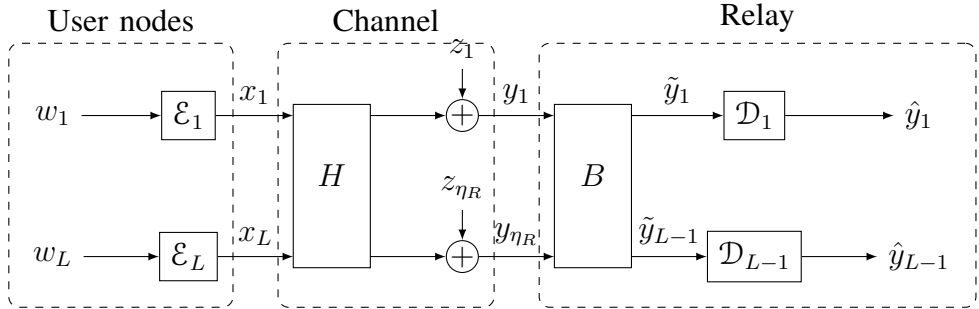

\subsection{Channel Model}
\label{sec:channel-model}

We assume the relay is equipped with $\eta_R$ antennas and therefore receives $\eta_R$ signals $y_1,\dots,y_{\eta_R}$.
The channel is characterized by the channel matrix $H \in \setRealNumbers^{\eta_R \times L}$ whose entries $h_{ij}$ are the fading coefficients from the $j$-th user to the $i$-th antenna at the relay.
Note that by this representation the $i$-th row of $H$, denoted by $h_i$, represents the channel from user $i$ to the relay.
The channel model for the first phase is then given by 
\begin{equation}
  Y = H X + Z,
\end{equation}
where $X \in \setRealNumbers^{L \times n}$ is a matrix whose $i$-th row is the transpose transmit vector $\transpose{x}_i$ of user $i$.
Further, $Y \in \setRealNumbers^{\eta_R \times n}$ is a matrix whose $i$-th row represents the received data stream $\transpose{y}_i$ at the $i$-th antenna and $Z \in \setRealNumbers^{\eta_R \times n}$ is white Gaussian noise.
The rows of $Z$ are denoted by $z_i$ and are i.i.d. according to a normal distribution with zero mean and unit variance, i.e., $z_i \sim \mathcal{N}(0,\idmat{n})$.
The relay decodes $L-1$ linear combinations of the original messages and encodes them with a capacity achieving code.
The $L-1$ codewords are then sent to all users simultaneously in the remaining $L-1$ phases.
Note that for these phases we have a broadcast channel and we assume reciprocal channels which are constant over all $L$ phases.
Therefore the rate constraints for the last $L-1$ phases are given by the capacity of the individual point-to-point channels.
Due to the uplink-downlink duality for reciprocal channels and equal power constraints \cite[Chapter 10.3]{Tse2005}, these are always larger or equal to the \gls{MAC} rate region and the first phase will be the limiting phase.
Therefore we will focus on that phase for developing the achievable rate of the whole system.

\subsection{Relay Strategy}

The relay uses compute-and-forward \cite{Nazer2011} as relaying strategy and tries to decode $L-1$ linear combination of the original messages as shown in Figure \ref{fig:system-model-simo-case}.
It uses a preprocessing matrix $B$ to get the optimal signals prior to decoding.
The achievable computation rate is then given by \cite{Nazer2011,Zhan2010a}
\begin{equation}
  R(H,a_\ell,b_\ell) = \frac{1}{2} \log_2 \left( \frac{P}{\norm{b_\ell}^2 + P \norm{\transpose{H} b_\ell - a_\ell}^2 } \right),
  \label{eq:computation-rate-simo}
\end{equation}
where $a_\ell$ is the coefficient vector for the $\ell$-th linear combination and $b_\ell$ is the preprocessing vector corresponding to the $\ell$-th row in $B$.
The optimal preprocessing matrix for a given coefficient matrix $A = \transpose{(\transpose{a}_1,\dots,\transpose{a_{L-1}})}$ and the resulting rates have been found in \cite{Zhan2010a}.
We use these results and provide them in the following for completeness.

The optimal preprocessing matrix is given by
\begin{equation}
  B = A \transpose{H} (H \transpose{H} + \tfrac{1}{P} \idmat{\eta_R})^{-1}.
  \label{eq:optimal-preprocessing-matrix}
\end{equation}
Plugging in \eqref{eq:optimal-preprocessing-matrix} in \eqref{eq:computation-rate-simo}, we get
\begin{equation}
  R(H,a_\ell) = -\frac{1}{2} \log_2 (\transpose{a}_\ell VD\transpose{V} a_\ell),
\end{equation}
where $V \in \setRealNumbers^{L \times L}$ is the right eigenmatrix of $H$ and $D \in \setRealNumbers^{L \times L}$ is a diagonal matrix with elements
\begin{equation}
  D_{ii} = \begin{cases}
    \tfrac{1}{P \lambda_i + 1} & i \leq \rank(H) \\
    1 & i > \rank(H)
  \end{cases},
\end{equation}
where $\lambda_i$ is the $i$-th eigenvalue of $\transpose{H} H$.

All $L-1$ linear combinations have to be decodable at the relay to allow all $L$ users to decode all original messages.
Therefore the resulting achievable rate of all $L$ users is
\begin{equation}
  R_{CF} = \min_{\ell \in \{1,\dots,L-1\}} R(H,a_\ell).
\end{equation}
To get the highest possible rates we need to find a set of full rank equations with coefficient matrix $A \in \setIntegers^{L-1 \times L}$.
Because all $L$ users need to be able to decode all messages, we additionally require the rows of $A$ to be linear independent of all vectors $e_\ell$ where $e_\ell$ is the unit vector with a one at the $\ell$-th position and zeros elsewhere.
This results in the following optimization problem
\begin{equation}
  \label{eq:optimize-coefficient-matrix}
  \begin{split}
    &\max_{A_\ell}~ \min_{\ell \in \{1,\dots,L-1\}} \left( -\frac{1}{2} \log_2 (\transpose{a}_\ell V D \transpose{V} a_\ell )\right)\\
    &\text{s.t.}~\rank(A_\ell) = L,~ \forall i \in \{1,\dots,L\}
  \end{split}  
\end{equation}
where
\begin{equation*}
  A_\ell =
  \begin{pmatrix}
    \transpose{a}_1\\
    \vdots\\
    \transpose{a}_{L-1}\\
    \transpose{e}_\ell
  \end{pmatrix}
  =
  \begin{pmatrix}
    a_{1,1} & a_{1,2} & \dots & a_{1,L} \\
    \vdots & \vdots & \ddots & \vdots \\
    a_{L-1,1} & a_{L-1,2} & \dots & a_{L-1,L}\\
    0 & 0 & 1 & 0
  \end{pmatrix}.
\end{equation*}
This can be solved by several algorithms \cite{Richter2012,Wei2012} with additional constraints. 

For the first phase we get the following rate constraint
\begin{equation}
  R_s + R_d \leq R_{CF}.
\end{equation}
In the second phase the relay maps all decoded linear combinations $\hat{y}_i \in \Lambda_F \cap \mathcal{V}_C$ to an index of the set $\{1,2,\dots,2^{\lfloor nR_r \rfloor} \}$ and uses a capacity achieving code to encode and an optimal beamforming vector to transmit to the user nodes with rate 
\begin{equation}
  \label{eq:transmission-rate-relay}
  R_r = \max_{\norm{\omega}^2 \leq 1}~ \min_{\ell \in \{1,\dots,L\}} \tfrac{1}{2} \log_2 (1 + P (\transpose{\omega} h_\ell)^2 ),
\end{equation}
where $\omega$ is the multicast beamforming vector at the relay.
An efficient way to obtain the optimal beamforming vector can be found in \cite{Mochaourab2011}.
% TODO: result for \eta_R \geq L?

\subsection{Decoding at the user nodes}

In each of the $L-1$ phases, each user node receives an index of the set $\{1,2,\dots,2^{\lfloor nR_r \rfloor }\}$.
They can decode as long as the transmission rate from relay to user is less than the point-to-point capacity of the channels, which is given in \eqref{eq:transmission-rate-relay}.
If decoded successfully they know the lattice point $\hat{y}_{\ell} \in \Lambda_F \cap \mathcal{V}_C$ being transmitted by the relay.

User $\ell$ can decode the original lattice points from the other users by solving the following system of linear equations
\begin{equation}
  A_\ell T = \tilde{Y}_\ell
\end{equation}
where
\begin{equation}  
  \tilde{Y}_\ell =
  \begin{pmatrix}
    \hat{y}_{1,1} & \hat{y}_{1,2} & \dots & \hat{y}_{1,n} \\
    \vdots & \vdots & \ddots & \vdots \\
    \hat{y}_{L-1,1} & \hat{y}_{L-1,2} & \dots & \hat{y}_{L-1,n} \\
    t_{\ell,1} & t_{\ell,2} & \dots & t_{\ell,n}
  \end{pmatrix}.
\end{equation}
It gets an estimate of the transmitted lattice points $t_\ell$ of the users $1,\dots,L$ by
\begin{equation}
  \hat{T_\ell} = A_\ell^{-1} \tilde{Y}_\ell.
\end{equation}
User $\ell$ already knows $t_\ell$ because this is its own message.
Please note that the users get the lattice points without the dither because the relay already subtracted it.

After solving for $T_\ell$ each user knows all $L$ lattice point.
If the lattice point is known, the message and the bin is known.
Therefore each user gets all $L$ messages.

% By applying the modulo operation with respect to $\Lambda_B$ to every decoded lattice point $\hat{x}_i,~i \in \{1,\dots,L\}$ it gets the lattice points $\hat{t}_i \in \mathcal{L}$ and maps them back to the message set $\Field_p^k$ with $\phi^{-1}(\hat{t}_i)$, where
% \begin{equation}
%   \phi^{-1}(t) = (\transpose{G} G)^{-1} \transpose{G} g^{-1} (p [G_B^{-1} t] \bmod \setIntegers^n), \, t \in \Lambda_F,
% \end{equation}
% which is the inverse function of \cref{eq:message-to-lattice}.
% This results in the message estimates $\hat{w}_i$, i.e.,\ 
% \begin{equation}
%   \hat{w}_i = \phi^{-1} \left( x_i \bmod \Lambda_B \right),
% \end{equation}
% for all $i \in \{1,\dots,L\}$.

From all $L$ phases we obtain a rate constraint for the source nodes given by
\begin{equation}
  \label{eq:constraint-rt}
  R_d + R_s \leq \min \{R_{CF}, R_r\} = R_{CF}.
\end{equation}
From the lattice code construction in Section \ref{sec:nested-lattice-code} and \eqref{eq:constraint-rt} we get the following constraint for the secure communication rate $R_s$
\begin{align}
  R_s \leq R_{CF} - R_d.
\end{align}
We get the same constraint for all source nodes because all use the same nested lattice chain.
To ensure that the relay will not get any information about individual messages, the rate $R_d$ has to be chosen appropriately.
This will be addressed in the next section.

\subsection{Achievable Secrecy Rate Region}

In this section we provide an achievable secrecy rate region.
The proof is given in the appendix.

\begin{theorem}[Achievable Secrecy Rate]
  \label{theorem:secrecy-rate-region-simo-l-user}
  Consider a multi-way relay channel with $L$ users and $\eta_R$ antennas at the relay.
  The channel is characterized by the matrix $H$ whose entries $h_{ij}$ represent the fading coefficient from the $j$-th user to the $i$-th antenna at the relay.
  All nodes can only communicate via the relay and have no direct links.
  Each node has a transmit power constraint $\norm{x_i}^2 \leq nP,~\forall i \in \{1,\dots,L,r\}$.
  Then the weak secrecy rate region is given by
  \begin{equation*}
    \label{eq:secrecy-rate-region-simo-l-user}
    L \cdot R_s \leq \max \left\{0, L R_{CF} - \tfrac{1}{2} \log_2 \det (\idmat{\eta_R} + P H \transpose{H}) \right\}
  \end{equation*}
  where 
  \begin{equation*}
    R_{CF} = \min_{i \in \{1,\dots,L-1\}} R(H,a_i)
  \end{equation*}
  with $a_i$ chosen in programming problem \eqref{eq:optimize-coefficient-matrix}.
\end{theorem}

\subsection{Power Allocation}

Every source node has a power constraint $\norm{x_\ell} \leq nP$.
However, it is not always optimal to send at full power and therefore the transmit power for each source node needs to be optimized.
We define a diagonal matrix $P_{\text{tr}} = \diag(\sqrt{P_1},\dots,\sqrt{P_L})$ which contains the square roots of the individual transmit powers of the source nodes.
Further, we define the effective channel as $\tilde{H} = H P_{\text{tr}}$.
This results in the following optimization problem
\begin{equation}
  \max_{P_\ell \leq P} L R_{\text{CF}} - \tfrac{1}{2} \log_2 \det( \idmat{\eta_R} + HP_{\text{tr}} \transpose{P}_{\text{tr}}\transpose{H} ),
\end{equation}
where
\begin{equation}
  R_{CF} = \min_{\ell \in \{1,\dots,L-1\}} -\tfrac{1}{2} \log_2(\transpose{a}_\ell V D \transpose{V} a_\ell),
\end{equation}
where $V$ is the right singular matrix of $\tilde{H}$ and $D$ is a diagonal matrix with elements
\begin{equation}
  D_{ii} = \begin{cases}
    \tfrac{1}{\lambda_i + 1} & i \leq \rank(\tilde{H}) \\
    1 & i > \rank(\tilde{H})
  \end{cases},
\end{equation}
where $\lambda_i$ is the $i$-th eigenvalue of $\transpose{\tilde{H}}\tilde{H}$.

This problem is hard to solve because of the non-linearity and non-convexity.
The scope of this paper does not include an analytic result or an algorithm.
For the simulations we use a grid search algorithm to obtain the optimal power allocation.

\section{SISO and MISO Multi-Way Relay Channel}
\label{sec:siso-miso-l-user}

In this section we provide the achievable secrecy rates for the following cases:
\begin{itemize}
\item single antenna at the source nodes and the relay; \gls{SISO},
\item multiple antennas at the source nodes and single antenna at the relay; \gls{MISO}.
\end{itemize}
These can be written as special cases of the \gls{SIMO} case.

\subsection{SISO}

The results for the \gls{SISO} case can be derived directly from Theorem \ref{theorem:secrecy-rate-region-simo-l-user} where the channel matrix reduces to a vector.
\begin{corollary}
  \label{corollary:secrecy-rate-region-siso-l-user}
  Consider a multi-way relay channel with $L$ users and single antennas at all nodes.
  The channel from user $\ell$ to the relay is characterized by the coefficient $h_\ell \in \setRealNumbers$ with $h = \transpose{(h_1,\dots,h_L)}$.
  All nodes can only communicate via the relay and have no direct links.
  Each node has a transmit power constraint $\norm{x_i}^2 \leq nP,~\forall i \in \{1,\dots,L,r\}$.
  Then the weak secrecy rate region is given by
  \begin{equation*}
    \label{eq:secrecy-rate-region-siso-l-user}
    L \cdot R_s \leq \max \left\{0, L R_{CF} - \tfrac{1}{2} \log_2 (1 + \norm{\tilde{h}}^2) \right\}
  \end{equation*}
  where
  \begin{equation*}
    R_{CF} = \min_{i \in \{1,\dots,L-1\}} \log_2^+ \bigg( \bigg( \norm{a_i}^2 - \frac{ (\transpose{\tilde{h}} a_i)^2 }{ 1 + \norm{\tilde{h}}^2 } \bigg)^{-1} \bigg).
  \end{equation*}
  Further, $\tilde{h} = \diag(\sqrt{P_1},\dots,\sqrt{P_L}) \cdot h$ is the effective channel and $P_\ell \leq P$ is the transmit power at user $\ell$.
\end{corollary}

\subsection{MISO}

For the \gls{MISO} case we assume no cooperation at the source nodes for choosing the beamforming vectors.
Therefore it is optimal to use \gls{MRT} in the direction of the channel.
This leaves us with an reduced optimization problem where we only have to choose the optimal power allocation.
The effective channel is
\begin{equation}
  \tilde{h} = \diag(\sqrt{P_1},\dots,\sqrt{P_L}) \cdot \transpose{(\transpose{h}_1 \omega_1,\dots,\transpose{h}_L \omega_L)},
\end{equation}
where $h_\ell$ is the channel vector from user $\ell$ to the relay and $\omega_\ell$ is the beamforming vector with $\norm{\omega_\ell} \leq 1$ for user $\ell$.
Using \gls{MRT} we get
\begin{equation}
  \omega_\ell = \tfrac{h_\ell}{\norm{h_\ell}} \text{ and } \transpose{h}_\ell \omega_\ell = \norm{h_\ell}. 
\end{equation}
The secrecy rate is now equivalent to the one in the \gls{SISO} case and we get the following corollary.
\begin{corollary}
  \label{corollary:secrecy-rate-region-miso-l-user}
  Consider a multi-way relay channel with $L$ users and $\eta_T$ antennas at the user nodes.
  The relay is equipped with a single antenna.
  The channel from user $\ell$ to the relay is characterized by the vector $h_\ell \in \setRealNumbers^{\eta_T}$.
  All nodes can only communicate via the relay and have no direct links.
  Each node has a transmit power constraint $\norm{x_i}^2 \leq nP,~\forall i \in \{1,\dots,L,r\}$.
  Then the weak secrecy rate region is given by
  \begin{equation*}
    \label{eq:secrecy-rate-region-miso-l-user}
    L \cdot R_s \leq \max \left\{0, L R_{CF} - \tfrac{1}{2} \log_2 (1 + \norm{\tilde{h}}^2) \right\}
  \end{equation*}
  where 
  \begin{equation*}
    R_{CF} = \min_{i \in \{1,\dots,L-1\}} \log_2^+ \bigg( \bigg( \norm{a_i}^2 - \frac{ (\transpose{\tilde{h}} a_i)^2 }{ 1 + \norm{\tilde{h}}^2 } \bigg)^{-1} \bigg).
  \end{equation*}
  Further, $\tilde{h} = \diag(\sqrt{P_1},\dots,\sqrt{P_L}) \cdot \transpose{(\norm{h_1},\dots,\norm{h_L})}$ is the effective channel and $P_\ell \leq P$ is the transmit power at user $\ell$.
\end{corollary}

  % Consider a multi-way relay channel with $L$ users and $\eta_R$ antennas at the relay.
  % The channel is characterized by the matrix $H$ whose entries $h_{ij}$ represent the fading coefficient from the $j$-th user to the $i$-th antenna at the relay.
  % All nodes can only communicate via the relay and have no direct links.
  % Each node has a transmit power constraint $\norm{x}^2 \leq nP$.
  % Then the weak secrecy rate region is given by
  % \begin{equation*}
  %   \label{eq:secrecy-rate-region-simo-l-user}
  %   L \cdot R_s \leq \max \left\{0, L R_{CF} - \tfrac{1}{2} \log_2 \det (\idmat{\eta_R} + P H \transpose{H}) \right\}
  % \end{equation*}
  % where 
  % \begin{equation*}
  %   R_{CF} = \min_{\ell \in \{1,\dots,L-1\}} R(H,a_\ell)
  % \end{equation*}
  % with $a_\ell$ chosen in programming problem \cref{eq:optimize-coefficient-matrix}.

  % \begin{align}
  %   &L \cdot R_s \leq \min_{\ell \in \{1,\dots,L-1\}} \max \Big\{ 0, \\
  %   &\log_2^+ \bigg( \bigg( \norm{a_\ell}^2 - \frac{ (\transpose{\tilde{h}} a_\ell)^2 }{ 1 + \norm{\tilde{h}}^2 } \bigg)^{-1} \bigg) - \frac{1}{2} \log_2 (1 + \norm{\tilde{h}}^2 ) \bigg\} \nonumber
  % \end{align}
  % where $\tilde{h} = \diag(\sqrt{P_1},\dots,\sqrt{P_L}) \cdot \transpose{(\norm{h_1},\dots,\norm{h_L})}$ is the effective channel and $h_i$ are the channel vectors from node $i$ to the relay.

\section{2-User Case: Two-Way Relay Channel}
\label{sec:2-user-case}

In this section we consider the special case of two users which is also known as the two-way relay channel.
We restrict ourselves to the SISO case where all nodes are equipped with single antennas.
% TODO: references
This system model is quite common in literature and one can find a lot of work on the achievable rate without secrecy \cite{Wilson2010,Zhang2009,Popovski2007} as well as with secrecy \cite{Tekin2008,Pierrot2011,He2013a,ElGamal2013,Bloch2013}.
Most of the work models the two-way relay channel as a wiretap channel where the second user is helping the first user by jamming the eavesdropper, i.e., the relay \cite{He2008,He2014}.
This differs from our work because we assume that both users transmit a secure message simultaneously.
Further, most of the work so far considers the two-way relay channel without fading and those results cannot be directly extended to fading channels.

However, to be able to compare our result to existing schemes we relax our model to match the one assumed in \cite{He2008}.
That is, user 1 transmits a secure message to user 2 via an untrusted relay.
User 2 helps by sending a jamming signal to the relay.
The main difference to our more general model is the fact, that the message of the second user does not need to be secure.
Further, we assume that the channel coefficients are equal and set to 1.
With this relaxed model we get the following secrecy rate.
\begin{theorem}
  \label{th:secrecy-rate-2-user-wiretap}
  Consider a two-way relay channel with fading coefficients $h_1 = h_2 = 1$ and equal power constraint $P$.
  The following secrecy rate is achievable with a cooperative jammer
  \begin{align*}
    R_s &\leq \max \left\{ 0, \frac{1}{2} \log_2 \left( \frac{1}{2} + P \right) - \frac{1}{2} \log_2 \left( 1+ \frac{P}{1+P} \right) \right\}\\
    &= R_s^{\text{CF}}.
  \end{align*}
\end{theorem}

\begin{proof}
  The following proof does not necessarily require that the channel fading coefficients are 1.
  Therefore we provide the proof with fading coefficients and choose them to be 1 at the end to get result of Theorem \ref{th:secrecy-rate-2-user-wiretap}.
  The proof follows the same steps as the proof for the $L$-user multi-way relay channel, except that we have only the following secrecy requirement:
  \begin{equation}
    \lim_{n \to \infty} \tfrac{1}{n} H(W_1 \mid Y_r) \geq \lim_{n \to \infty} \tfrac{1}{n} H(W_1) = R_s. 
  \end{equation}
  We start by bounding the conditional entropy of the message $W_1$ given the received signal $Y_r$ at the relay,
  \begin{align}
    &H(W_1 \mid Y_r) \\
    &= H(W_1 \mid X_1,Y_r) + H(X_1 \mid Y_r) - H(X_1 \mid W_1,Y_r) \label{eq:proof-secrecy-rate-2-user-fano}\\
    &\geq H(X_1 \mid Y_r) - n \delta(n) \\
    &= H(X_1 \mid Y_r) - H(X_1) + H(X_1) - n \delta(n) \\
    &= H(X_1) - I(X_1 ; Y_r) - n \delta(n),
  \end{align}
  where we used Fano's inequality to bound the last term in \eqref{eq:proof-secrecy-rate-2-user-fano}.
  This is because the size of each bin is kept small enough such that given $W_1$, the eavesdropper/the relay can determine $X_1$ from its received signal $Y_r$ \cite{He2008}.
  
  We now need a lower bound on the mutual information between $X_1$ and $Y_r$.
  \begin{align}
    &I(X_1 ; Y_r) \\
    &= h(Y_r) - h(Y_r \mid X_1) \\
    &\leq n \cdot \tfrac{1}{2} \log_2 ( 2 \pi e (P h_1^2 + P h_2^2 + 1)) - h(h_2 X_2 + Z_r) \label{eq:proof-secrecy-rate-2-user-mutual-information}
  \end{align}
  The last inequality follows from the fact that a normal distribution maximizes the entropy under an average power constraint.
  The last term can be expressed as follows,
  \begin{align}
    &h(h_2 X_2 + Z_r) \\
    &= h(h_2 X_2 + Z_r \mid h_2 X_2) + I(h_2 X_2 ; h_2 X_2 + Z_r) \\
    &= h(Z_r) + I(h_2 X_2 ; h_2 X_2 + Z_r), \label{eq:proof-secrecy-rate-2-user-entropy}
  \end{align}
  where $h(Z_r) = \frac{1}{2} \log_2 ( 2 \pi e )$.
  From \cite{Erez2004} we know that
  \begin{equation}
    \label{eq:proof-secrecy-rate-2-user-erez-zamir-bound}
    \lim_{n \to \infty} \tfrac{1}{n} I(h_2 X_2 ; h_2 X_2 + Z_r) = \tfrac{1}{2} \log_2 ( 1 + h_2^2 P ).
  \end{equation}
  If we combine \eqref{eq:proof-secrecy-rate-2-user-mutual-information}, \eqref{eq:proof-secrecy-rate-2-user-entropy} and \eqref{eq:proof-secrecy-rate-2-user-erez-zamir-bound} we get
  \begin{align}
    &\lim_{n \to \infty} \tfrac{1}{n} H(X_1 \mid Y_r) \\
    &\geq \lim_{n \to \infty} \tfrac{1}{n} \left[ H(X_1) - I(X_1 ; Y_r) - n \delta(n) \right] \\
    &\geq \lim_{n \to \infty} \tfrac{1}{n} [ H(X_1) - n \cdot \tfrac{1}{2} \log_2 ( 2 \pi e (h_1^2 P + h_2^2 P + 1)) \nonumber\\
    &\hspace*{5mm}+ h(Z_r) + I(h_2 X_2 ; h_2 X_2 + Z_r) ]\\
    &= (R_s + R_d) - \tfrac{1}{2} \log_2 (h_1^2 P + h_2^2 P + 1) \nonumber\\
    &\hspace*{5mm}+ \lim_{n \to \infty} \tfrac{1}{n} I(h_2 X_2 ; h_2 X_2 + Z_r) \\
    &= (R_s + R_d) - \tfrac{1}{2} \log_2 \left( 1 + \frac{h_1^2 P}{1 + h_2^2 P} \right)
  \end{align}
  Because $R_s + R_d \leq R_{\text{CF}}$ we get the following result,
  \begin{equation}
    R_s \leq R_{\text{CF}} - \tfrac{1}{2} \log_2 \left( 1 + \frac{h_1^2 P}{1 + h_2^2 P} \right).
  \end{equation}
  Choosing $h_1 = h_2 = 1$ we get the result of Theorem \ref{th:secrecy-rate-2-user-wiretap}.
  This concludes the proof.
\end{proof}

% We investigated the two-way relay channel in a previous work \cite{Richter2013} and provide here only the result, which is needed for later investigation and comparison. 

% \begin{corollary}[Achievable Secrecy Rate]
%   \label{theorem:secrecy-rate-region-2-user}
%   Consider a two-way relay-channel with fading coefficients $h_1$ and $h_2$ and $h = \transpose{(h_1,h_2)}$.
%   Each node has a transmit power constraint $\norm{x_i}^2 \leq nP,~\forall i \in \{1,2,r\}$.
%   With $\tilde{h} = \transpose{(\sqrt{P_1} h_1, \sqrt{P_2} h_2)}$ and $P_i \leq P$, the weak secrecy rate region is given by
%   \begin{equation*}
%     \label{eq:secrecy-rate-region}
%     2 R_s \leq \max \left\{0, 2 R_{CF} - \tfrac{1}{2} \log_2 (1 + P_1 h_1^2 + P_2 h_2^2) \right\}
%   \end{equation*}
%   where 
%   \begin{equation*}
%     R_{CF} = \max_{\substack{a \in \setIntegers\\a_1 \neq 0\\a_2 \neq 0}} \frac{1}{2} \log_2^+ \left( \left( \norm{a}^2 - \frac{(\transpose{\tilde{h}}a)^2}{1+\norm{\tilde{h}}^2} \right)^{-1} \right).
%   \end{equation*} 
% \end{corollary}

In the following we compare our result for the two-way relay channel to other approaches and illustrate the results in Figure \ref{fig:power-rate-region}.
The first result is a scheme by He and Yener which provides a weak secrecy result \cite{He2008}, i.e.,
\begin{equation}
  R_s \leq \max\{0, \tfrac{1}{2} \log_2 ( \tfrac{1}{2} + P ) - 1 \} = R_s^{\text{HS}}.
\end{equation}
This result is extended in \cite{He2013a} for strong secrecy.
The achievable secrecy rate is shown to be the same as for weak secrecy by utilizing a hash function.
Please note that we can extend our result as well with a hash function to get a result for strong secrecy.
This will be the topic of future research and is not considered in this paper.

The second scheme is provided by Kashyap et al. in \cite{Kashyap2012}.
The achievable secrecy rates are
\begin{equation}
  R_s \leq \max\{0, \tfrac{1}{2} \log_2 (P) - 1 - \log_2(e)\} = R_s^{\text{KP}}
\end{equation}
and
\begin{equation}
  R_s \leq \max\{0, \tfrac{1}{2} \log_2 ( \tfrac{1}{2} + P ) - \log_2 (2 e) \} = R_s^{\text{KS}}
\end{equation}
for perfect and strong secrecy, respectively.
Observe that $R_{\text{CF}} \geq R_s^{\text{CF}} \geq R_s^{\text{HS}} \geq R_s^{\text{KS}} \geq R_s^{\text{KP}}$.

% \begin{enumerate}
% \item Our scheme: The point $(P,\max \{0, \tfrac{1}{4} \log_2 (1+ \tfrac{2P}{\sigma^2}) - \tfrac{1}{2} \})$ in the power-rate region is achievable with \emph{weak} secrecy.
% \item He et al.\ \cite{He2013a}: The point $(P,\max \{0, \tfrac{1}{4} \log_2 (\tfrac{1}{2} + \tfrac{P}{\sigma^2}) - \tfrac{1}{2} \})$ in the power-rate region is achievable with \emph{strong} secrecy.
% \item Kashyap et al.\ \cite{Kashyap2012}: The point $(P,\tfrac{1}{2} \log_2 (P/(4e^2\sigma^2)))$ in the power-rate region is achievable with \emph{perfect} secrecy.
% \end{enumerate}
% Because the results of He et al. only consider one user to transmit a message while the other user sends a cooperative jamming signal, we assume time sharing for the comparison.
%
\begin{figure}
  \centering
  \tikzsetnextfilename{power-rate-region}

\begin{tikzpicture}[>=latex]

  \tikzpicturedependsonfile{figures/power-rate-region.tikz}

  \pgfplotstableread{./figures/power-rate-region.csv}\powerrateregion
  
  \begin{axis} [
    grid=major,
    xmin=0,
    xmax=30,
    ymin=0,
    ymax=6,
    xtick={0,2,...,30},
    ytick={0,0.2,...,6.0},
    legend style={
      cells={anchor=west},
      legend pos=north west,
    },
    tick label style={
      font=\tiny
    },
    ylabel={Achievable secrecy rate $R_s$ [bit/cu]},
    xlabel={$P/\sigma^2$ [dB]},
    ]

    \addplot[draw=black,semithick] table[x=snr,y=nazer] {\powerrateregion};
    \addlegendentry{\footnotesize Nazer and Gastpar (no secrecy)}

    \addplot[draw=black,mark=triangle,mark repeat=20,semithick] table[x=snr,y=richter] {\powerrateregion};
    \addlegendentry{\footnotesize Our scheme (weak secrecy)}

    \addplot[draw=red,dashdotted,semithick] table[x=snr,y=he] {\powerrateregion};
    \addlegendentry{\footnotesize Scheme by He et al. (strong secrecy)}
    
    \addplot[draw=blue,dashed,semithick] table[x=snr,y=kashyap_strong] {\powerrateregion};
    \addlegendentry{\footnotesize Scheme by Kashyap et al. (strong secrecy)}

    \addplot[draw=brown,mark=o,mark repeat=20,semithick] table[x=snr,y=kashyap_perfect] {\powerrateregion};
    \addlegendentry{\footnotesize Scheme by Kashyap et al. (perfect secrecy)}

  \end{axis}

\end{tikzpicture}
  \caption{Achievable secrecy rate of the two-way relay channel with $h=\transpose{(1,1)}$ and $\sigma^2 = 1$ for different schemes.}
  \label{fig:power-rate-region}
\end{figure}
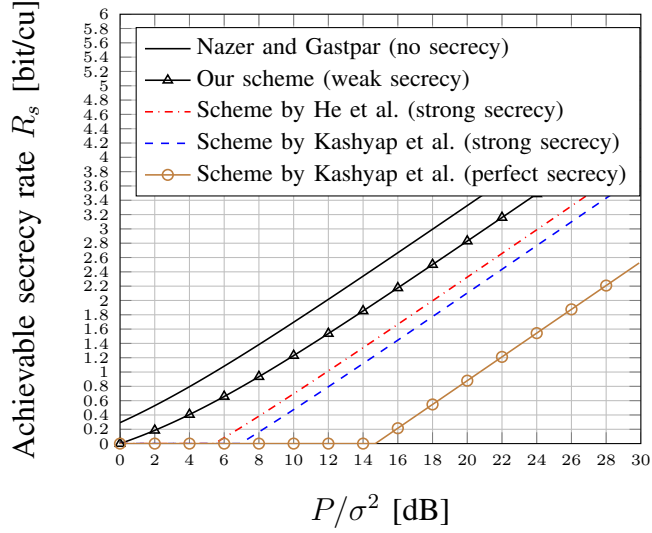
%
% We compare all three schemes for $\sigma^2=1$ in Figure \ref{fig:power-rate-region}.
% It is interesting to note that the scheme of Kashyap et al.\ achieves a higher perfect secrecy rate for large transmit power $P$ than the weak secrecy scheme proposed in this paper.
% However, for low SNR the latter performs better than the former.

\section{Discussion}
\label{sec:discussion}

In this section we discuss and illustrate Theorem \ref{theorem:secrecy-rate-region-simo-l-user}, Corollary \ref{corollary:secrecy-rate-region-siso-l-user} and Corollary \ref{corollary:secrecy-rate-region-miso-l-user}.
It is interesting to note, that the achievable secrecy rate is the difference between the achievable compute-and-forward sum rate and the sum capacity of the \gls{MAC}.
This means, we get a secrecy rate greater than zero if the compute-and-forward rate region is larger than the \gls{MAC} capacity region.
This is illustrated in Figure \ref{fig:mac-region-0.9-0.8-10} for the two user case, where we plotted the achievable rate regions for the channel coefficients $\transpose{(0.9,0.8)}$ and a transmit power constraint of $P = 10$ Watt.
The dots in Figure \ref{fig:mac-region-0.9-0.8-10} mark the corner points of the achievable compute-and-forward rate regions for different linear combinations.
The solid line illustrates the border of the \gls{MAC} capacity region.
\begin{figure}
  \centering
  \tikzsetnextfilename{mac-region-0.9-0.8-10}
\begin{tikzpicture}[>=latex,scale=0.8]

  \tikzpicturedependsonfile{figures/mac-region-0.9-0.8-10.tikz}

  \draw[color=black!30,xstep=0.502218149204879,ystep=0.502218149204879] (0,0) grid (8.53770853648295,7.53327223807319);

  \tikzset{tick-node/.style={fill=white,inner sep=1pt}}
  \draw[<->,thick] (0,8.0) |- (8.9,0);

  \node[rotate=90] at (-1.0,4.0) {$R_2$ [bit/cu]};
  \node[] at (4.45,-1.0) {$R_1$ [bit/cu]};

  \draw (0.502218149204879,2pt) -- (0.502218149204879,-2pt) node[tick-node,anchor=north,yshift=-1pt] {\tiny 0.1};
  \draw (1.00443629840976,2pt) -- (1.00443629840976,-2pt) node[tick-node,anchor=north,yshift=-1pt] {\tiny 0.2};
  \draw (1.50665444761464,2pt) -- (1.50665444761464,-2pt) node[tick-node,anchor=north,yshift=-1pt] {\tiny 0.3};
  \draw (2.00887259681952,2pt) -- (2.00887259681952,-2pt) node[tick-node,anchor=north,yshift=-1pt] {\tiny 0.4};
  \draw (2.51109074602440,2pt) -- (2.51109074602440,-2pt) node[tick-node,anchor=north,yshift=-1pt] {\tiny 0.5};
  \draw (3.01330889522927,2pt) -- (3.01330889522927,-2pt) node[tick-node,anchor=north,yshift=-1pt] {\tiny 0.6};
  \draw (3.51552704443415,2pt) -- (3.51552704443415,-2pt) node[tick-node,anchor=north,yshift=-1pt] {\tiny 0.7};
  \draw (4.01774519363903,2pt) -- (4.01774519363903,-2pt) node[tick-node,anchor=north,yshift=-1pt] {\tiny 0.8};
  \draw (4.51996334284391,2pt) -- (4.51996334284391,-2pt) node[tick-node,anchor=north,yshift=-1pt] {\tiny 0.9};
  \draw (5.02218149204879,2pt) -- (5.02218149204879,-2pt) node[tick-node,anchor=north,yshift=-1pt] {\tiny 1.0};
  \draw (5.52439964125367,2pt) -- (5.52439964125367,-2pt) node[tick-node,anchor=north,yshift=-1pt] {\tiny 1.1};
  \draw (6.02661779045855,2pt) -- (6.02661779045855,-2pt) node[tick-node,anchor=north,yshift=-1pt] {\tiny 1.2};
  \draw (6.52883593966343,2pt) -- (6.52883593966343,-2pt) node[tick-node,anchor=north,yshift=-1pt] {\tiny 1.3};
  \draw (7.03105408886831,2pt) -- (7.03105408886831,-2pt) node[tick-node,anchor=north,yshift=-1pt] {\tiny 1.4};
  \draw (7.53327223807319,2pt) -- (7.53327223807319,-2pt) node[tick-node,anchor=north,yshift=-1pt] {\tiny 1.5};
  \draw (8.03549038727807,2pt) -- (8.03549038727807,-2pt) node[tick-node,anchor=north,yshift=-1pt] {\tiny 1.6};
  \draw (8.53770853648295,2pt) -- (8.53770853648295,-2pt) node[tick-node,anchor=north,yshift=-1pt] {\tiny 1.7};

  \draw (2pt,0.502218149204879) -- (-2pt,0.502218149204879) node[tick-node,anchor=east,xshift=-1pt] {\tiny 0.1};
  \draw (2pt,1.00443629840976) -- (-2pt,1.00443629840976) node[tick-node,anchor=east,xshift=-1pt] {\tiny 0.2};
  \draw (2pt,1.50665444761464) -- (-2pt,1.50665444761464) node[tick-node,anchor=east,xshift=-1pt] {\tiny 0.3};
  \draw (2pt,2.00887259681952) -- (-2pt,2.00887259681952) node[tick-node,anchor=east,xshift=-1pt] {\tiny 0.4};
  \draw (2pt,2.51109074602440) -- (-2pt,2.51109074602440) node[tick-node,anchor=east,xshift=-1pt] {\tiny 0.5};
  \draw (2pt,3.01330889522927) -- (-2pt,3.01330889522927) node[tick-node,anchor=east,xshift=-1pt] {\tiny 0.6};
  \draw (2pt,3.51552704443415) -- (-2pt,3.51552704443415) node[tick-node,anchor=east,xshift=-1pt] {\tiny 0.7};
  \draw (2pt,4.01774519363903) -- (-2pt,4.01774519363903) node[tick-node,anchor=east,xshift=-1pt] {\tiny 0.8};
  \draw (2pt,4.51996334284391) -- (-2pt,4.51996334284391) node[tick-node,anchor=east,xshift=-1pt] {\tiny 0.9};
  \draw (2pt,5.02218149204879) -- (-2pt,5.02218149204879) node[tick-node,anchor=east,xshift=-1pt] {\tiny 1.0};
  \draw (2pt,5.52439964125367) -- (-2pt,5.52439964125367) node[tick-node,anchor=east,xshift=-1pt] {\tiny 1.1};
  \draw (2pt,6.02661779045855) -- (-2pt,6.02661779045855) node[tick-node,anchor=east,xshift=-1pt] {\tiny 1.2};
  \draw (2pt,6.52883593966343) -- (-2pt,6.52883593966343) node[tick-node,anchor=east,xshift=-1pt] {\tiny 1.3};
  \draw (2pt,7.03105408886831) -- (-2pt,7.03105408886831) node[tick-node,anchor=east,xshift=-1pt] {\tiny 1.4};
  \draw (2pt,7.53327223807319) -- (-2pt,7.53327223807319) node[tick-node,anchor=east,xshift=-1pt] {\tiny 1.5};

  \path[draw,blue] (0,7.2508) -- (2.6785,7.2508) -- (8.0000,1.9293) -- (8.0000,0);

  \node[circle,anchor=center,fill=red,minimum size=0.3em,inner sep=0,label=right:{\tiny (1, 2)}] at (0.1188,0.1188) {};
  \node[circle,anchor=center,fill=red,minimum size=0.3em,inner sep=0,label=right:{\tiny (0, 1)}] at (0.0000,1.9293) {};
  \node[circle,anchor=center,fill=red,minimum size=0.3em,inner sep=0,label=above:{\tiny (1, 0)}] at (2.6785,0.0000) {};

  \path[draw,dashed] (0,7.2415) -- (7.2415,7.2415) -- (7.2415,0);
  \node[circle,anchor=center,fill=red,minimum size=0.3em,inner sep=0,label=right:{\tiny (1, 1)}] at (7.2415,7.2415) {};

  \path[draw,dashed] (0,1.6241) -- (1.6241,1.6241) -- (1.6241,0);
  \node[circle,anchor=center,fill=red,minimum size=0.3em,inner sep=0,label=right:{\tiny (2, 1)}] at (1.6241,1.6241) {};

  \path[draw,dashed] (0,2.2193) -- (2.2193,2.2193) -- (2.2193,0);
  \node[circle,anchor=center,fill=red,minimum size=0.3em,inner sep=0,label=right:{\tiny (2, 2)}] at (2.2193,2.2193) {};

  \draw[<->] ($(2.6785,7.2508)!(7.2415,7.2415)!(8.0000,1.9293)$) -- node[above left] {$R_s$} (7.2415,7.2415);

\end{tikzpicture}
  \caption{MAC capacity region (blue line) with achievable compute-and-forward rates for different coefficient vectors $a$ (red dots) for $h = \transpose{(0.9,0.8)}$ and $P=10$ Watt.}
  \label{fig:mac-region-0.9-0.8-10}
\end{figure}
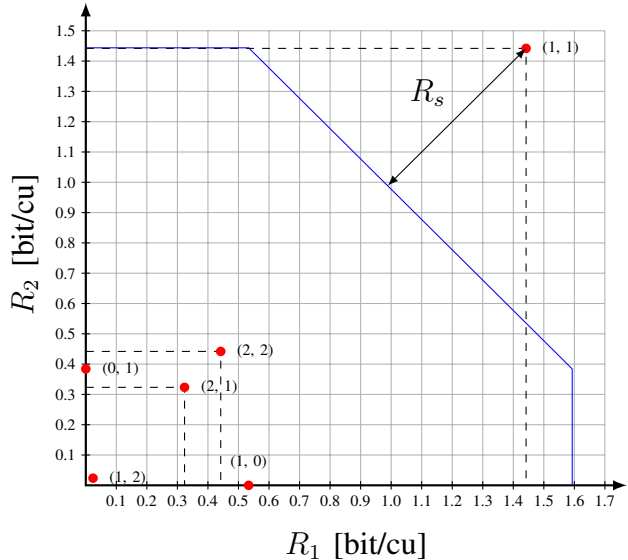
As one can see, a rate greater than the \gls{MAC} sum capacity is only achievable for a single coefficient vector, in this example $a = \transpose{(1,1)}$.
In general there can be at most one point or several linear dependent points outside the \gls{MAC} capacity region.
Otherwise the relay would be able to decode the single messages by solving a system of linear equations which contradicts the \gls{MAC} capacity definition.
From this illustration we see, that it is only possible to transmit secure messages via the relay, if the source nodes transmit at a rate outside the \gls{MAC} capacity.
This ensures that the relay cannot decode the single messages but only a linear combination, if the compute-and-forward rate is higher than the \gls{MAC} sum capacity.
In Figure \ref{fig:mac-region-0.9-0.8-10} it is optimal to transmit at full power.
In Figure \ref{fig:mac-region-0.8-0.5-10} we show that this is not always the case.
Therefore we plotted the achievable rate regions for the channel coefficients $\transpose{(0.8,0.5)}$ and a transmit power constraint of $P=10$ Watt.
\begin{figure}
  \centering
  \tikzsetnextfilename{mac-region-0.8-0.5-10-power-allocation}

\begin{tikzpicture}[>=latex,scale=0.8]

  \tikzpicturedependsonfile{figures/mac-region-0.8-0.5-10-power-allocation.tikz}

  \draw[color=black!30,xstep=0.554107704688271,ystep=0.554107704688271] (0,0) grid (8.31161557032407,5.54107704688271); 

  \tikzset{tick-node/.style={fill=white,inner sep=1pt}}
  \draw[<->,thick] (0,6.0) |- (8.8,0);

  \node[rotate=90] at (-1.0,3.0) {$R_2$ [bit/cu]};
  \node[] at (4.4,-1.0) {$R_1$ [bit/cu]};

\draw (0.554107704688271,2pt) -- (0.554107704688271,-2pt) node[tick-node,anchor=north,yshift=-1pt] {\tiny 0.1};
\draw (1.10821540937654,2pt) -- (1.10821540937654,-2pt) node[tick-node,anchor=north,yshift=-1pt] {\tiny 0.2};
\draw (1.66232311406481,2pt) -- (1.66232311406481,-2pt) node[tick-node,anchor=north,yshift=-1pt] {\tiny 0.3};
\draw (2.21643081875309,2pt) -- (2.21643081875309,-2pt) node[tick-node,anchor=north,yshift=-1pt] {\tiny 0.4};
\draw (2.77053852344136,2pt) -- (2.77053852344136,-2pt) node[tick-node,anchor=north,yshift=-1pt] {\tiny 0.5};
\draw (3.32464622812963,2pt) -- (3.32464622812963,-2pt) node[tick-node,anchor=north,yshift=-1pt] {\tiny 0.6};
\draw (3.87875393281790,2pt) -- (3.87875393281790,-2pt) node[tick-node,anchor=north,yshift=-1pt] {\tiny 0.7};
\draw (4.43286163750617,2pt) -- (4.43286163750617,-2pt) node[tick-node,anchor=north,yshift=-1pt] {\tiny 0.8};
\draw (4.98696934219444,2pt) -- (4.98696934219444,-2pt) node[tick-node,anchor=north,yshift=-1pt] {\tiny 0.9};
\draw (5.54107704688271,2pt) -- (5.54107704688271,-2pt) node[tick-node,anchor=north,yshift=-1pt] {\tiny 1};
\draw (6.09518475157098,2pt) -- (6.09518475157098,-2pt) node[tick-node,anchor=north,yshift=-1pt] {\tiny 1.1};
\draw (6.64929245625925,2pt) -- (6.64929245625925,-2pt) node[tick-node,anchor=north,yshift=-1pt] {\tiny 1.2};
\draw (7.20340016094753,2pt) -- (7.20340016094753,-2pt) node[tick-node,anchor=north,yshift=-1pt] {\tiny 1.3};
\draw (7.75750786563580,2pt) -- (7.75750786563580,-2pt) node[tick-node,anchor=north,yshift=-1pt] {\tiny 1.4};
\draw (8.31161557032407,2pt) -- (8.31161557032407,-2pt) node[tick-node,anchor=north,yshift=-1pt] {\tiny 1.5};

\draw (2pt,0.554107704688271) -- (-2pt,0.554107704688271) node[tick-node,anchor=east,xshift=-1pt] {\tiny 0.1};
\draw (2pt,1.10821540937654) -- (-2pt,1.10821540937654) node[tick-node,anchor=east,xshift=-1pt] {\tiny 0.2};
\draw (2pt,1.66232311406481) -- (-2pt,1.66232311406481) node[tick-node,anchor=east,xshift=-1pt] {\tiny 0.3};
\draw (2pt,2.21643081875309) -- (-2pt,2.21643081875309) node[tick-node,anchor=east,xshift=-1pt] {\tiny 0.4};
\draw (2pt,2.77053852344136) -- (-2pt,2.77053852344136) node[tick-node,anchor=east,xshift=-1pt] {\tiny 0.5};
\draw (2pt,3.32464622812963) -- (-2pt,3.32464622812963) node[tick-node,anchor=east,xshift=-1pt] {\tiny 0.6};
\draw (2pt,3.87875393281790) -- (-2pt,3.87875393281790) node[tick-node,anchor=east,xshift=-1pt] {\tiny 0.7};
\draw (2pt,4.43286163750617) -- (-2pt,4.43286163750617) node[tick-node,anchor=east,xshift=-1pt] {\tiny 0.8};
\draw (2pt,4.98696934219444) -- (-2pt,4.98696934219444) node[tick-node,anchor=east,xshift=-1pt] {\tiny 0.9};
\draw (2pt,5.54107704688271) -- (-2pt,5.54107704688271) node[tick-node,anchor=east,xshift=-1pt] {\tiny 1};

\path[draw,dashdotted] (0,5.0073) -- (4.1560,5.0073) -- (8.0000,1.1634) -- (8.0000,0);

\node[cross out,anchor=center,draw=black,minimum size=0.3em,inner sep=0,label=right:{\tiny (2, 1)}] at (2.4228,2.4228) {};
\node[cross out,anchor=center,draw=black,minimum size=0.3em,inner sep=0,label=above:{\tiny (1, 0)}] at (4.1560,0.0000) {};
\node[cross out,anchor=center,draw=black,minimum size=0.3em,inner sep=0,label=left:{\tiny (1, 1)}] at (4.9077,4.9077) {};
\node[cross out,anchor=center,draw=black,minimum size=0.3em,inner sep=0,label=right:{\tiny (0, 1)}] at (0.0000,1.1634) {};

\path[draw,blue] (0,5.0073) -- (3.5126,5.0073) -- (7.1135,1.4064) -- (7.1135,0);

\node[circle,anchor=center,fill=red,minimum size=0.3em,inner sep=0,label=right:{\tiny (0, 1)}] at (0.0000,1.4064) {};
\node[circle,anchor=center,fill=red,minimum size=0.3em,inner sep=0,label=above:{\tiny (1, 0)}] at (3.5126,0.0000) {};
\node[circle,anchor=center,fill=red,minimum size=0.3em,inner sep=0,label=right:{\tiny (1, 1)}] at (5.0073,5.0073) {};
\node[circle,anchor=center,fill=red,minimum size=0.3em,inner sep=0,label=right:{\tiny (2, 1)}] at (1.4335,1.4335) {};

\end{tikzpicture}
  \caption{MAC capacity region with achievable compute-and-forward rates for different coefficient vectors $a$ at full transmit power (dashdotted line and crosses) and at optimal transmit power $(7.7, 10.0)$ Watt (blue solid line and red dots) for $h = \transpose{(0.8,0.5)}$ and $P=10$ Watt.}
  \label{fig:mac-region-0.8-0.5-10}
\end{figure}
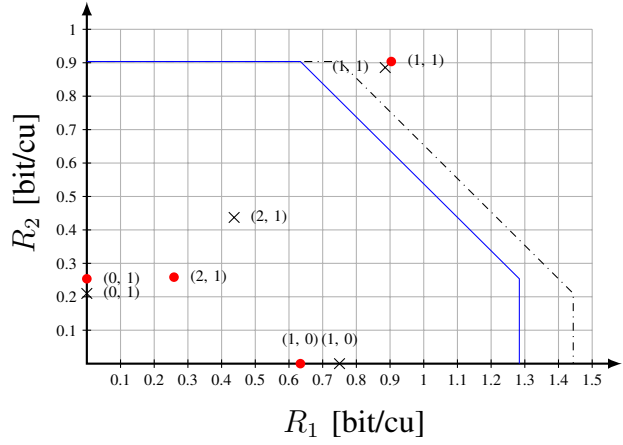
As one can see optimizing the power means reducing the transmit power for one user node.
Because the channel coefficients of the effective channel are then closer to each other, we achieve a higher computation rate for the coefficient $a = \transpose{(1,1)}$ while reducing the \gls{MAC} sum capacity at the same time.
We want to stress here the not common behavior that it is possible to increase the secrecy rate by reducing the transmit power.

\subsection{Achievability Of Positive Secrecy Rates}

The achievable compute-and-forward sum rate is the double of the computation rate at the relay.
The achievable computation rate depends highly on the channel coefficients because the compute-and-forward framework tries to approximate the real valued channel coefficients with integer valued network coding coefficients.
This means that there does not always exist a network coding coefficient vector with an achievable computation rate which is larger than the \gls{MAC} capacity.
We show in Figure \ref{fig:two-user-channel-realizations-5dB} and Figure \ref{fig:two-user-channel-realizations-20dB} the achievable secrecy rate for different channel realizations in a two user scenario, i.e., the two-way relay channel with a single antenna at all nodes.
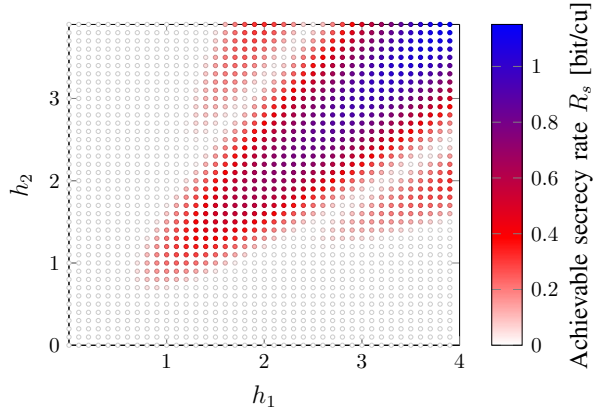
\begin{figure}
  \centering
  \tikzsetnextfilename{two-user-channel-realizations-5dB}

\begin{tikzpicture}[scale=0.8,>=latex]

  \tikzpicturedependsonfile{./figures/two-user-channel-realizations-5dB.tikz}

  \pgfplotsset{
    colormap={zeroone}{rgb(0)=(1,0,0); rgb(1)=(0,0,1)},
    colormap={custom01}{rgb(0)=(1,1,1); rgb(1)=(1,0,0); rgb(3)=(0,0,1)}
  }

  \pgfplotstableread{./figures/two-user-channel-realizations-5dB.csv}\twouserchannel
  
  \begin{axis} [
    xtick={1,2,...,4},
    xmin=0,
    xmax=4,
    width=8cm,
    legend style={
      cells={anchor=west},
      legend pos=north west,
    },
    tick label style={
      font=\small
    },
    view={0}{90},
    scatter,
    colorbar,
    colorbar style={
      colormap name=custom01,
      ylabel={Achievable secrecy rate $R_s$ [bit/cu]},
      ylabel style={yshift=-2.7cm},
    },
    % colorbar sampled line,
    ylabel={$h_2$},
    xlabel={$h_1$},
    % x label style={
    % at={(current axis.right of origin)},
    % anchor=north west  
    % }
    ]

    % \addplot3[only marks,mark options={mark size=1pt}] table[x=h1,y=h2,z=outside] {\twouserchannel};
    \addplot3[only marks,mark options={mark size=1pt}] table[x=h1,y=h2,z=Rs] {\twouserchannel};

  \end{axis}

  % \begin{axis}[
  %   ytick={0.1,0.2,...,0.5},
  %   yshift=-2.45cm,
  %   height=4cm,
  %   width=8cm,
  %   xmin=0,
  %   xmax=4,
  %   ymin=0,
  %   ymax=0.6,
  %   axis y line=left,
  %   hide x axis,
  %   ylabel={$\Pr(h_1)$},
  %   xlabel={$h_1$},
  %   tick label style={
  %     font=\small
  %   },
  %   y dir=reverse,
  %   y label style={
  %     font=\small,
  %     % at={(current axis.below origin)},
  %     % anchor=south
  %   }
  %   ]
    
  %   \addplot[mark=none] table[x=h1,y=Prh1] {\twouserchannel};

  % \end{axis}

  % \begin{axis}[
  %   xshift=-40,
  %   width=3cm,
  %   height=8cm,
  %   ymin=0,
  %   ymax=1,
  %   xmin=0,
  %   axis x line=bottom,
  %   hide y axis,
  %   xlabel={$\Pr(h_2)$},
  %   ylabel={$h_2$},
  %   tick label style={
  %     font=\small
  %   },
  %   ]
    
  %   \addplot[mark=none] table[x=Prh2,y=h2] {\twouserchannel};

  % \end{axis}

\end{tikzpicture}
  \caption{Existence of a compute-and-forward rate tuple outside of the multiple access channel capacity region. $P/\sigma^2 = 5$dB.}
  %  including the probability for a channel realization, where $h_i \sim \mathcal{N}(0,1)$
  \label{fig:two-user-channel-realizations-5dB}
\end{figure}
\begin{figure}
  \centering
  \tikzsetnextfilename{two-user-channel-realizations-20dB}

\begin{tikzpicture}[scale=0.8,>=latex]

  \tikzpicturedependsonfile{./figures/two-user-channel-realizations-20dB.tikz}

  \pgfplotsset{
    colormap={custom01}{rgb(0)=(1,1,1); rgb(1)=(1,0,0); rgb(3)=(0,0,1)}
  }

  \pgfplotstableread{./figures/two-user-channel-realizations-20dB.csv}\twouserchannel
  
  \begin{axis} [
    xtick={1,2,...,4},
    xmax=4,
    xmin=0,
    width=8cm,
    legend style={
      cells={anchor=west},
      legend pos=north west,
    },
    tick label style={
      font=\small
    },
    view={0}{90},
    scatter,
    colorbar,
    colorbar style={
      colormap name=custom01,
      ylabel={Achievable secrecy rate $R_s$ [bit/cu]},
      ylabel style={yshift=-2.7cm},
    },
    % colorbar sampled line,
    ylabel={$h_2$},
    xlabel={$h_1$},
    % x label style={
    % at={(current axis.right of origin)},
    % anchor=north west  
    % }
    ]

    % \addplot3[only marks,mark options={mark size=1pt}] table[x=h1,y=h2,z=outside] {\twouserchannel};
    \addplot3[only marks,mark options={mark size=1pt}] table[x=h1,y=h2,z=Rs] {\twouserchannel};

  \end{axis}

  % \begin{axis}[
  %   ytick={0.1,0.2,...,0.5},
  %   yshift=-2.45cm,
  %   height=4cm,
  %   width=8cm,
  %   xmin=0,
  %   xmax=4,
  %   ymin=0,
  %   ymax=0.6,
  %   axis y line=left,
  %   hide x axis,
  %   ylabel={$\Pr(h_1)$},
  %   xlabel={$h_1$},
  %   tick label style={
  %     font=\small
  %   },
  %   y dir=reverse,
  %   y label style={
  %     font=\small,
  %     % at={(current axis.below origin)},
  %     % anchor=south
  %   }
  %   ]
    
  %   \addplot[mark=none] table[x=h1,y=Prh1] {\twouserchannel};

  % \end{axis}

  % \begin{axis}[
  %   xshift=-40,
  %   width=3cm,
  %   height=8cm,
  %   ymin=0,
  %   ymax=1,
  %   xmin=0,
  %   axis x line=bottom,
  %   hide y axis,
  %   xlabel={$\Pr(h_2)$},
  %   ylabel={$h_2$},
  %   tick label style={
  %     font=\small
  %   },
  %   ]
    
  %   \addplot[mark=none] table[x=Prh2,y=h2] {\twouserchannel};

  % \end{axis}

\end{tikzpicture}
  \caption{Existence of a compute-and-forward rate tuple outside of the multiple access channel capacity region. $P/\sigma^2 = 20$dB.}
  %  including the probability for a channel realization, where $h_i \sim \mathcal{N}(0,1)$
  \label{fig:two-user-channel-realizations-20dB}
\end{figure}
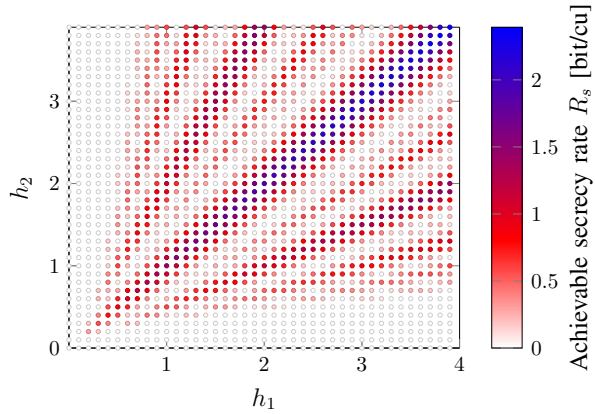
One can see that the achievable secrecy rate reaches its highest values if the channel coefficients are equal.
One can also see that small channel coefficient values do not achieve positive secrecy rates.
One can compensate this behavior in part by adjusting the power allocation such that the effective channel coefficients are close to each other. 
Unfortunately, when we assume that the channel coefficients are distributed by $h_i \sim \mathcal{N}(0,1)$, the small channel coefficients occur with higher probability.
The question arises how many channel realizations result in a positive secrecy rate when we draw the channel coefficients from a normal distribution with zero mean and unit variance.
The result for the \gls{SISO} case without optimized power allocation is show in Figure \ref{fig:channel-histogram-siso} where we used 10,000 i.i.d. channel realizations.
One can see that only in the two-way relay channel we achieve positive secrecy rates for a reasonable amount of channel realizations.
This might be a depressing result but we can do better by optimizing the transmit power and introducing multiple antennas.
The result for a \gls{SNR} of 5dB is shown in Figure \ref{fig:channel-histogram-all}.
One can see that alone by optimizing the power allocation in the \gls{SISO} case we increase the channel realizations resulting in a positive secrecy rate from 18.19\% to 52.1\%.
Introducing multiple antennas at the source nodes results in 100\% of positive secrecy rates.
Introducing multiple antennas at the relay is contra-beneficial because we give the eavesdropper more degrees of freedom.

\begin{figure}
  \centering
  \tikzsetnextfilename{l-user-channel-histogram}

\begin{tikzpicture}[>=latex]
  
  \tikzpicturedependsonfile{./figures/l-user-channel-histogram.tikz}

  \pgfplotstableread{./figures/l-user-channel-histogram-5dB.csv}\kuserchanneli
  \pgfplotstableread{./figures/l-user-channel-histogram-10dB.csv}\kuserchannelii
  \pgfplotstableread{./figures/l-user-channel-histogram-15dB.csv}\kuserchanneliii
  \pgfplotstableread{./figures/l-user-channel-histogram-20dB.csv}\kuserchanneliv

  \begin{axis}[
    enlarge y limits=0.3,
    xbar,
    bar width=5pt,
    xmin=0,
    xmax=65,
    ymin=2,
    ymax=4,
    height=5cm,
    width=8cm,
    ytick={2,3,4},
    legend style={
      cells={anchor=west},
      legend pos=north east,
    },
    ylabel={\small Number of users},
    xlabel={\small \% channel realizations with $R_s > 0$},
    nodes near coords={\tiny \pgfmathprintnumber{\pgfplotspointmeta}},
    nodes near coords align=horizontal,
    ]
    
    \addplot table[y=dimension,x=decodable_percent] {\kuserchanneli};
    \addlegendentry{\small $\tfrac{P}{\sigma^2} = 5$dB}

    \addplot table[y=dimension,x=decodable_percent] {\kuserchannelii};
    \addlegendentry{\small $\tfrac{P}{\sigma^2} = 10$dB}

    \addplot table[y=dimension,x=decodable_percent] {\kuserchanneliii};
    \addlegendentry{\small $\tfrac{P}{\sigma^2} = 15$dB}

    \addplot table[y=dimension,x=decodable_percent] {\kuserchanneliv};
    \addlegendentry{\small $\tfrac{P}{\sigma^2} = 20$dB}

  \end{axis}

\end{tikzpicture}
  \caption{Percentage of channel realizations resulting in positive secrecy rates for different numbers of users and \gls{SNR} values for the \gls{SISO} multi-way relay channel without optimal power allocation. 10,000 channel realizations are drawn from a normal distribution $\mathcal{N}(0,1)$.}
  \label{fig:channel-histogram-siso}
\end{figure}
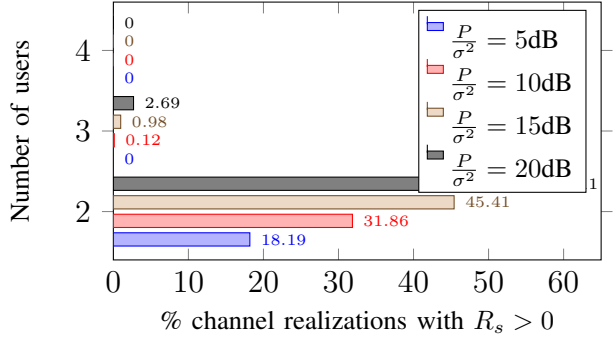

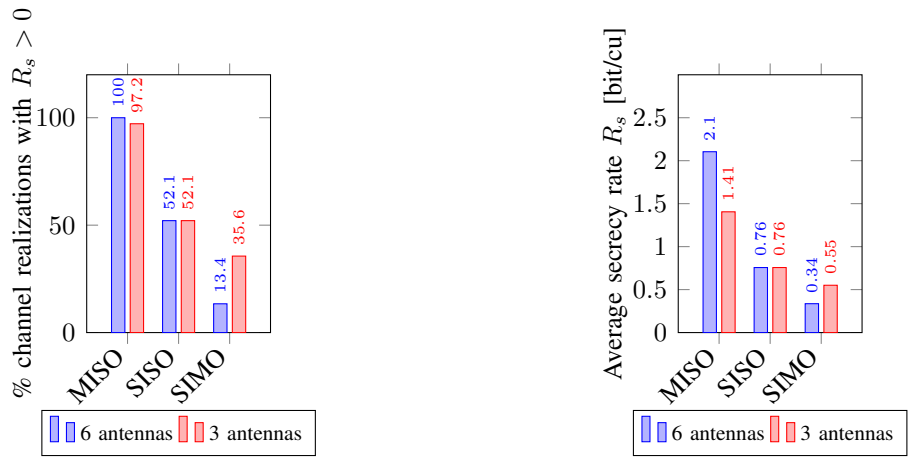
\begin{figure}
  \centering
  \tikzsetnextfilename{2-user-channel-histogram}

\begin{subfigure}[t]{0.47\linewidth}
  \centering
  \begin{tikzpicture}[>=latex]
    
    \tikzpicturedependsonfile{./figures/2-user-channel-histogram.tikz}

    \pgfplotstableread{./figures/2-user-channel-histogram-5dB-6-antennas.csv}\kuserchanneli
    \pgfplotstableread{./figures/2-user-channel-histogram-5dB-3-antennas.csv}\kuserchannelii
    % \pgfplotstableread{./figures/l-user-channel-histogram-15dB.csv}\kuserchanneliii
    % \pgfplotstableread{./figures/l-user-channel-histogram-20dB.csv}\kuserchanneliv

    \begin{axis}[
      enlarge x limits=0.4,
      enlarge y limits={upper,value=0.2},
      ybar,
      ymin=0,
      ymax=100,
      xtick={},
      xticklabels={},
      xtick={1,2,3},
      xticklabels={MISO,SISO,SIMO},
      x tick label style={
        font=\small,
        rotate=45,
        anchor=east,
      },
      y tick label style={
        font=\small,
      },
      bar width=5pt,
      height=5cm,
      width=4.0cm,
      legend style={
        at={(0.5,-0.3)},
        anchor=north,
        legend columns=-1,
        cells={
          anchor=west,
          font=\scriptsize,
        },
        % legend pos=north west,
      },
      ylabel={\small \% channel realizations with $R_s > 0$},
      ylabel style={yshift=-0.4cm},
      nodes near coords={\tiny \pgfmathprintnumber{\pgfplotspointmeta}},
      nodes near coords align=vertical,
      every node near coord/.append style={rotate=90, anchor=west}
      ]
      
      \addplot table[x=type,y=decodable_percent] {\kuserchanneli};
      \addplot table[x=type,y=decodable_percent] {\kuserchannelii};

      \legend{6 antennas, 3 antennas}

    \end{axis}

  \end{tikzpicture}
  \caption{Percentage of channel realizations resulting in positive secrecy rates}
\end{subfigure}
%
% \quad
\tikzsetnextfilename{2-user-channel-histogram-rate}
\begin{subfigure}[t]{0.47\linewidth}
  \centering
  \begin{tikzpicture}[>=latex]
    
    \tikzpicturedependsonfile{./figures/2-user-channel-histogram.tikz}

    \pgfplotstableread{./figures/2-user-channel-histogram-5dB-6-antennas.csv}\kuserchanneli
    \pgfplotstableread{./figures/2-user-channel-histogram-5dB-3-antennas.csv}\kuserchannelii

    \begin{axis}[
      enlarge x limits=0.4,
      ybar,
      ymin=0,
      ymax=3,
      ytick={0,0.5,...,2.5},
      xtick={1,2,3},
      xticklabels={MISO,SISO,SIMO},
      x tick label style={
        font=\small,
        rotate=45,
        anchor=east,
      },
      y tick label style={
        font=\small,
      },
      bar width=5pt,
      height=5cm,
      width=4.0cm,
      legend style={
        at={(0.5,-0.3)},
        anchor=north,
        legend columns=-1,
        cells={
          anchor=west,
          font=\scriptsize,
        },
      },
      ylabel={\small Average secrecy rate $R_s$ [bit/cu]},
      ylabel style={yshift=-0.4cm},
      nodes near coords={\tiny \pgfmathprintnumber{\pgfplotspointmeta}},
      nodes near coords align=vertical,
      every node near coord/.append style={rotate=90, anchor=west}
      ]
      
      \addplot table[x=type,y=avg_secrecy_rate] {\kuserchanneli};
      \addplot table[x=type,y=avg_secrecy_rate] {\kuserchannelii};

      \legend{6 antennas, 3 antennas}

    \end{axis}

  \end{tikzpicture}
  \caption{Average secrecy rate over all positive rates}
\end{subfigure}
  \caption{Percentage of channel realizations resulting in positive secrecy rates and the according average secrecy rates for $P/\sigma^2 = 5$dB. All schemes use optimal power allocation. 1000 channel realizations are drawn from a normal distribution $\mathcal{N}(0,1)$.}
  \label{fig:channel-histogram-all}
\end{figure}

\section{Conclusion}
\label{sec:conclusion}

In this paper we have presented an achievable secrecy rate region for the L-user relay channel where all nodes want to securely transmit messages via an untrusted relay.
We have shown that the proposed coding strategy, based on the compute-and-forward framework, supports simultaneous secure communications.
The proposed achievable secrecy rate is the difference between the \gls{MAC} sum capacity and the computation rate of the compute-and-forward framework.
We have provided a proof for the achievability of the secrecy rate region and have discussed this result.
We have seen that this scheme performs quit good in the 2-user case, commonly known as the two-way untrusted relay channel.
Therefore we investigated this scenario in detail and showed the dependency of the performance on the channel realizations.
We showed that by introducing multiple antennas at the source nodes and optimizing the transmit power we can significantly increase the secrecy rate.

\section*{Acknowledgment}

The authors would like to thank Martin Mittelbach from the Department of Electrical Engineering and Information Technology at Technische Universität Dresden for valuable discussions.

\appendix

\label{sec:proof-secrcecy-rate}

\begin{table}
  \centering
  \caption{Overview of variables and symbols}
  \begin{tabularx}{\linewidth}{llX}
    \toprule
    variable & distribution & comment\\
    \hline
    $W_i$ & $\sim \mathcal{U}(\{1,2,\dots,\lceil 2^{nR_s} \rceil \})$ & secret message of length $k$ \\
    % $T_i$ & $\sim \mathcal{U}(\Lambda_F \cap \mathcal{V}_B)$ & encoded message in $\mathcal{L}$\\
    % $V_i$ & $\sim \mathcal{U}(\Lambda_B \cap \mathcal{V}_C)$ & defines the coset of $\mathcal{L}$\\
    $U_i$ & $\sim \mathcal{U}(\mathcal{V}_C)$ & dither at node $i$\\
    $X_i$ & $\sim \mathcal{U}(\mathcal{V}_C)$ & transmit vector at node $i$\\
    $Z_j$ & $\sim \mathcal{N}(0,\idmat{n})$ & AWGN at relay antenna $j$\\
    $Y_j$ & continuous & received vector at relay antenna $j$\\
    \bottomrule
  \end{tabularx}
  \label{tab:variable-overview}
\end{table}

In this section we provide the proof of Theorem \ref{theorem:secrecy-rate-region-simo-l-user}.
An overview over all random variables is provided in Table \ref{tab:variable-overview}.

\begin{proof}
  For the achievability of the secrecy rate region we must show that the following weak secrecy condition holds:
  \begin{align}
    \label{eq:weak-secrecy-condition-simo-l-user}
    \lim_{n \to \infty} \tfrac{1}{n} I(\mathcal{W}; \mathcal{Y} \mid U_1,\dots,U_L) &= 0,
  \end{align}
  for all $\mathcal{W} \in \mathcal{P}(\{W_1,\dots,W_L\})$ and all $\mathcal{Y} \in \mathcal{P}(\{Y_1,\dots,Y_{\eta_R}\})$ where $\mathcal{P}(X)$ is the power set of $X$.
  For the ease of readability we omit the condition on $U_1,\dots,U_L$ in the following since the dither is present and known in every mutual information and entropy and does not change the equations. 
  By using the chain rule for mutual information we see that%
  \footnote{Note that for $\ell = L: I(W_L ; Y_1,\dots,Y_{\eta_R} \mid W_1,\dots,W_{L-1}) \neq 0$ because of the random binning.}
  \begin{align}
    &I(W_1,\dots,W_L; Y_1,\dots,Y_{\eta_R}) - I(W_1,\dots,W_m; Y_1,\dots,Y_{\eta_R})\nonumber\\
    &\hspace*{1cm}= \sum_{\ell=m}^L I(W_\ell; Y_1,\dots,Y_{\eta_R} \mid W_1,\dots,W_{\ell-1}).
  \end{align}
  Because of the non-negativity of mutual information we have
  \begin{align}
    &I(W_1,\dots,W_L; Y_1,\dots,Y_{\eta_R})\nonumber\\ 
    &\hspace*{1cm}- I(W_1,\dots,W_m; Y_1,\dots,Y_{\eta_R}) \geq 0
  \end{align}
  and therefore 
  \begin{align}
    &I(W_1,\dots,W_L; Y_1,\dots,Y_{\eta_R})\nonumber\\
    &\hspace*{1cm}\geq I(W_1,\dots,W_m; Y_1,\dots,Y_{\eta_R}),~\forall m < L.
  \end{align}
  Using the symmetry of mutual information and the same arguments as above, we can show that
  \begin{align}
    &I(W_1,\dots,W_L; Y_1,\dots,Y_{\eta_R})\nonumber\\
    &\hspace*{1cm}\geq I(W_1,\dots,W_L; Y_1,\dots,Y_m),~\forall m < \eta_R.
  \end{align}
  Hence, it is sufficient to show that 
  \begin{equation}  
    \lim_{n \to \infty} \tfrac{1}{n} I(W_1,\dots,W_L; Y_1,\dots,Y_{\eta_R}) = 0
  \end{equation}
  which implies \eqref{eq:weak-secrecy-condition-simo-l-user}.
  This condition is equivalent to
  \begin{equation}
    \begin{split}
      &\lim_{n \to \infty} \tfrac{1}{n} H(W_1,\dots,W_L)\\
      \leq &\lim_{n \to \infty} \tfrac{1}{n} H(W_1,\dots,W_L \mid Y_1,\dots,Y_{\eta_R}).
    \end{split}
  \end{equation}
  We can explicitly write the left hand side and get
  \begin{equation}
    L \cdot R_{s} \leq \lim_{n \to \infty} \tfrac{1}{n} H(W_1,\dots,W_L \mid Y_1,\dots,Y_{\eta_R}).
  \end{equation}
  We now need a lower bound on the right hand side.
  \begin{align}
    & \hspace*{-4mm} \lim_{n \to \infty} \tfrac{1}{n} H(W_1,\dots,W_L \mid Y_1,\dots,Y_{\eta_R}) \\
    &= \lim_{n \to \infty} \tfrac{1}{n} [ H(W_1,\dots,W_L \mid X_1,\dots,X_L,Y_1,\dots,Y_{\eta_R}) \nonumber \\
    & \hspace*{11mm} + H(X_1,\dots,X_L \mid Y_1,\dots,Y_{\eta_R}) \\
    & \hspace*{11mm} - H(X_1,\dots,X_L \mid W_1,\dots,W_L,Y_1,\dots,Y_{\eta_R}) ] \nonumber \\
    &\overset{a)}{\geq} \lim_{n \to \infty} \tfrac{1}{n} [ H(X_1,\dots,X_L \mid Y_1,\dots,Y_{\eta_R}) - n \delta(n) ] \\
    &= \lim_{n \to \infty} \tfrac{1}{n} [ H(X_1,\dots,X_L \mid Y_1,\dots,Y_{\eta_R}) \\
    & \hspace*{11mm} - H(X_1,\dots,X_L) + H(X_1,\dots,X_L) - n \delta(n) ] \nonumber \\
    &= \lim_{n \to \infty} \tfrac{1}{n} [ H(X_1,\dots,X_L) \\
    & \hspace*{11mm} - I(X_1,\dots,X_L ; Y_1,\dots,Y_{\eta_R}) - n \delta(n) ] \nonumber \\
    &\overset{b)}{\geq} L \cdot (R_s + R_d) - \tfrac{1}{2} \log_2\det(\idmat{\eta_R} + PH\transpose{H})
  \end{align}
  We have used the following arguments:
  \begin{enumerate}[a)]
  \item We used Fano's inequality to bound the last term.
    This is because the size of each bin is kept small enough such that given $W_1,\dots,W_L$, the eavesdropper can determine $X_1,\dots,X_L$ from the received signals.
  \item We note that the term $I(X_1,\dots,X_L; Y_1,\dots,Y_{\eta_R})$ corresponds to the mutual information of a \gls{MIMO} channel.
    We rewrite the mutual information in terms of entropy, i.e.,
    \begin{align}
      &I(X_1,\dots,X_L; Y_1,\dots,Y_{\eta_R})\nonumber\\
      &= h(Y_1,\dots,Y_{\eta_R}) - h(Y_1,\dots,Y_{\eta_R} \mid X_1,\dots,X_L).
    \end{align}    
    % We note that this corresponds to a multiple-input multiple output (MIMO) channel and the capacity for a Gaussian channel is given for example in \cite{Telatar1999}.
    % For completeness we provide the basic steps in the following. 
    We use the fact that the normal distribution maximizes the entropy for an average power constraint to get an upper bound on the first term.
    Furthermore, from \cite[Section 3.2]{Telatar1999} we know that if $X$ is distributed according to a normal distribution with zero-mean and covariance $E[x\transpose{x}] = P \idmat{L}$ than $Y = H X + Z$ is also distributed according to a normal distribution with zero-mean and covarianz $E[y\transpose{y}] = P H\transpose{H} + \idmat{\eta_R}$.
    \begin{align*}
      h(Y_1,\dots,Y_{\eta_R}) &= h(Y_{1,1},\dots,Y_{1,n},\dots,Y_{\eta_R,1},\dots,Y_{\eta_R,n})\\
      &= h(Y_{1,1},\dots,Y_{\eta_R,1},\dots,Y_{1,n},\dots,Y_{\eta_R,n})\\
      %&= \sum_{i=1}^n h(Y_{1,i},\dots,Y_{\eta_R,i} \mid Y_{1,i-1},\dots,Y_{\eta_R,i-1},Y_{1,i-2},\dots,Y_{\eta_R,i-2},\dots,Y_{1,1},\dots,Y_{\eta_R,1})\\
      &\leq \sum_{i=1}^n h(Y_{1,i},\dots,Y_{\eta_R,i})\\
      &\leq n \cdot \tfrac{1}{2} \log_2 ((2 \pi e)^L \det (P H \transpose{H} + \idmat{\eta_R})).
    \end{align*}
    The only uncertainty in the received signals $Y_1,\dots,Y_{\eta_R}$, if the transmitted signals $X_1,\dots,X_L$ are given, comes from the noise $Z_1,\dots,Z_{\eta_R}$.
    The noise is i.i.d. with respect to a normal distribution which results in the following entropy: 
    \begin{align*}
      &h(Y_1,\dots,Y_{\eta_R} \mid X_1,\dots,X_L)\\
      &\hspace*{1cm}= h(Z_1,\dots,Z_{\eta_R})\\
      &\hspace*{1cm}= \sum_{i=1}^n h(Z_{1,i},\dots,Z_{\eta_R,i})\\
      &\hspace*{1cm}= n \cdot h(Z_{1,i},\dots,Z_{\eta_R,i})\\
      &\hspace*{1cm}= n \cdot \tfrac{1}{2} \log_2 ((2 \pi e)^L \det(\idmat{\eta_R})).
    \end{align*}
    Putting everything together results in
    \begin{align*}
      &I(X_1,\dots,X_L; Y_1,\dots,Y_{\eta_R})\\
      &\hspace*{1cm}\leq n \cdot \tfrac{1}{2} \log_2 ((2 \pi e)^L \det (P H \transpose{H} + \idmat{\eta_R}))\\
      &\hspace*{2cm}- n \cdot \tfrac{1}{2} \log_2 ((2 \pi e)^L \det(\idmat{\eta_R}))\\
      &\hspace*{1cm}= n \cdot \tfrac{1}{2} \log_2 \frac{(2 \pi e)^L \det (P H \transpose{H} + \idmat{\eta_R})}{(2 \pi e)^L \det(\idmat{\eta_R})}\\
      &\hspace*{1cm}= n \cdot \tfrac{1}{2} \log_2 \det( \idmat{\eta_R} + P H \transpose{H} ).
    \end{align*}
  \end{enumerate}
  Since all rates are symmetric and $R_s + R_d \leq R_{\text{CF}}$ we get the following weak secrecy rate
  \begin{equation}
    L \cdot R_s \leq L \cdot R_{\text{CF}} - \tfrac{1}{2} \log_2 \det( \idmat{\eta_R} + P H \transpose{H} ).
  \end{equation}
  This concludes the proof.
\end{proof}

\bibliographystyle{IEEEtran}
\bibliography{bibliography/references.bib}

\begin{thebibliography}{10}
\providecommand{\url}[1]{#1}
\csname url@samestyle\endcsname
\providecommand{\newblock}{\relax}
\providecommand{\bibinfo}[2]{#2}
\providecommand{\BIBentrySTDinterwordspacing}{\spaceskip=0pt\relax}
\providecommand{\BIBentryALTinterwordstretchfactor}{4}
\providecommand{\BIBentryALTinterwordspacing}{\spaceskip=\fontdimen2\font plus
\BIBentryALTinterwordstretchfactor\fontdimen3\font minus
  \fontdimen4\font\relax}
\providecommand{\BIBforeignlanguage}[2]{{%
\expandafter\ifx\csname l@#1\endcsname\relax
\typeout{** WARNING: IEEEtran.bst: No hyphenation pattern has been}%
\typeout{** loaded for the language `#1'. Using the pattern for}%
\typeout{** the default language instead.}%
\else
\language=\csname l@#1\endcsname
\fi
#2}}
\providecommand{\BIBdecl}{\relax}
\BIBdecl

\bibitem{Ahlswede2000}
R.~Ahlswede, N.~Cai, S.-Y.~R. Li, and R.~W. Yeung, ``Network information
  flow,'' \emph{IEEE Transactions on Information Theory}, vol.~46, no.~4, pp.
  1204--1216, 2000.

\bibitem{Fragouli2007}
C.~Fragouli and E.~Soljanin, \emph{Network Coding Fundamentals}.\hskip 1em plus
  0.5em minus 0.4em\relax Now Publishers, 2007.

\bibitem{Yeung2008}
R.~W. Yeung, \emph{Information Theory and Network Coding}.\hskip 1em plus 0.5em
  minus 0.4em\relax Springer, 2008.

\bibitem{Yeung2005}
R.~W. Yeung, S.-Y.~R. Li, N.~Cai, and Z.~Zhang, \emph{Network Coding Theory
  Part {I}: Single Source}.\hskip 1em plus 0.5em minus 0.4em\relax Now
  Publishers, 2005.

\bibitem{Yeung2005a}
------, \emph{Network Coding Theory Part {II}: Multiple Source}.\hskip 1em plus
  0.5em minus 0.4em\relax Now Publishers, 2005.

\bibitem{Katti2008}
S.~Katti, H.~Rahul, W.~Hu, D.~Katabi, M.~M\'{e}dard, and J.~Crowcroft, ``{XORs}
  in the air: Practical wireless network coding,'' \emph{IEEE/ACM Transactions
  on Networking}, vol.~16, no.~3, pp. 497--510, Jun. 2008.

\bibitem{Zhang2006}
S.~Zhang, S.~C. Liew, and P.~P. Lam, ``Hot topic: physical-layer network
  coding,'' in \emph{Proc. of the Annual International Conference on Mobile
  Computing and Networking}, 2006.

\bibitem{Nazer2011a}
B.~Nazer and M.~Gastpar, ``Reliable physical layer network coding,''
  \emph{Proceedings of the IEEE}, vol.~99, no.~99, pp. 438--460, 2011.

\bibitem{Nazer2011}
------, ``Compute-and-forward: Harnessing interference through structured
  codes,'' \emph{IEEE Transactions on Information Theory}, vol.~57, no.~10, pp.
  6463--6486, 2011.

\bibitem{Erez2004}
U.~Erez and R.~Zamir, ``Achieving {1/2 log(1 + SNR)} on the {AWGN} channel with
  lattice encoding and decoding,'' \emph{IEEE Transactions on Information
  Theory}, vol.~50, no.~10, pp. 2293--2314, 2004.

\bibitem{Wyner1975}
A.~D. Wyner, ``The wire-tap channel,'' \emph{Bell System Technical Journal},
  vol.~54, no.~8, pp. 1355--1387, 1975.

\bibitem{Csiszar1978}
I.~Csiszár and J.~Körner, ``Broadcast channels with confidential messages,''
  \emph{IEEE Transactions on Information Theory}, vol.~24, no.~3, pp. 339--348,
  1978.

\bibitem{Huang2013}
J.~Huang, A.~Mukherjee, and A.~L. Swindlehurst, ``Secure communication via an
  untrusted non-regenerative relay in fading channels,'' \emph{IEEE
  Transactions on Signal Processing}, vol.~61, no.~10, pp. 2536--2550, May
  2013.

\bibitem{He2010}
X.~He and A.~Yener, ``Cooperation with an untrusted relay: A secrecy
  perspective,'' \emph{IEEE Transactions on Information Theory}, vol.~56,
  no.~8, pp. 3807--3827, Aug. 2010.

\bibitem{He2013a}
------, ``Strong secrecy and reliable byzantine detection in the presence of an
  untrusted relay,'' \emph{IEEE Transactions on Information Theory}, vol.~59,
  no.~1, pp. 177--192, 2013.

\bibitem{He2013}
------, ``End-to-end secure multi-hop communication with untrusted relays,''
  \emph{IEEE Transactions on Wireless Communications}, vol.~12, no.~1, pp.
  1--11, Jan. 2013.

\bibitem{Tekin2008}
E.~Tekin and A.~Yener, ``The general {G}aussian multiple-access and two-way
  wiretap channels: Achievable rates and cooperative jamming,'' \emph{IEEE
  Transactions on Information Theory}, vol.~54, no.~6, pp. 2735--2751, 2008.

\bibitem{Tekin2010}
------, ``Correction to: The general {G}aussian multiple access and two-way
  wire-tap channels: Achievable rates and cooperative jamming,'' \emph{IEEE
  Transactions on Information Theory}, vol.~56, no.~9, pp. 4762--4763, 2010.

\bibitem{Pierrot2011}
A.~J. Pierrot and M.~R. Bloch, ``Strongly secure communications over the
  two-way wiretap channel,'' \emph{IEEE Transactions on Information Forensics
  and Security}, vol.~6, no.~3, pp. 595--605, 2011.

\bibitem{Choo2011}
L.-C. Choo, C.~Ling, and K.-K. Wong, ``Achievable rates for lattice coded
  {G}aussian wiretap channels,'' in \emph{IEEE International Conference on
  Communications Workshops (ICC)}, Jun. 2011, pp. 1--5.

\bibitem{Oggier2013}
\BIBentryALTinterwordspacing
F.~Oggier, P.~Solé, and J.-C. Belfiore, ``Lattice codes for the wiretap
  {G}aussian channel: Construction and analysis,'' Jan. 2013, submitted to IEEE
  Transactions on Information Theory. [Online]. Available:
  \url{http://arxiv.org/abs/1103.4086}
\BIBentrySTDinterwordspacing

\bibitem{Belfiore2010}
J.-C. Belfiore and F.~Oggier, ``Secrecy gain: A wiretap lattice code design,''
  in \emph{International Symposium on Information Theory and its Applications
  (ISITA)}, Oct. 2010, pp. 174--178.

\bibitem{Ling2012}
C.~Ling, L.~Luzzi, and J.-C. Belfiore, ``Lattice codes achieving strong secrecy
  over the mod-${\Lambda}$ {G}aussian channel,'' in \emph{Proc. of the
  International Symposium on Information Theory (ISIT)}, 2012, pp. 2306--2310.

\bibitem{Ling2012a}
\BIBentryALTinterwordspacing
C.~Ling, L.~Luzzi, J.-C. Belfiore, and D.~Stehlé, ``Semantically secure
  lattice codes for the {G}aussian wiretap channel,'' Oct. 2012. [Online].
  Available: \url{http://arxiv.org/abs/1210.6673}
\BIBentrySTDinterwordspacing

\bibitem{Belfiore2011}
J.-C. Belfiore, ``Lattice codes for the compute-and-forward protocol: The
  flatness factor,'' in \emph{IEEE Information Theory Workshop (ITW)}, Oct.
  2011, pp. 1--4.

\bibitem{Kashyap2012}
\BIBentryALTinterwordspacing
N.~Kashyap, S.~V, and A.~Thangaraj, ``Secure compute-and-forward in a
  bidirectional relay,'' 2012. [Online]. Available:
  \url{http://arxiv.org/abs/1206.3392}
\BIBentrySTDinterwordspacing

\bibitem{He2008}
X.~He and A.~Yener, ``Providing secrecy with lattice codes,'' in \emph{46th
  Annual Allerton Conference on Communication, Control, and Computing}, 2008,
  pp. 1199--1206.

\bibitem{Conway1999}
J.~H. Conway and N.~J.~A. Sloane, \emph{Sphere Packings, Lattices and Groups},
  3rd~ed.\hskip 1em plus 0.5em minus 0.4em\relax Springer, 1999.

\bibitem{Zamir2009}
R.~Zamir, ``Lattices are everywhere,'' in \emph{Information Theory and
  Applications Workshop}, 2009, pp. 392--421.

\bibitem{Erez2005a}
U.~Erez, S.~Litsyn, and R.~Zamir, ``Lattices which are good for (almost)
  everything,'' \emph{IEEE Transactions on Information Theory}, vol.~51,
  no.~10, pp. 3401--3416, 2005.

\bibitem{Tse2005}
D.~Tse and P.~Viswanath, \emph{Fundamentals of Wireless Communication}.\hskip
  1em plus 0.5em minus 0.4em\relax Cambridge University Press, 2005.

\bibitem{Zhan2010a}
J.~Zhan, B.~Nazer, U.~Erez, and M.~Gastpar, ``Integer-forcing linear
  receivers,'' in \emph{IEEE International Symposium on Information Theory
  Proceedings (ISIT)}.\hskip 1em plus 0.5em minus 0.4em\relax IEEE, 2010, pp.
  1022--1026.

\bibitem{Richter2012}
J.~Richter, C.~Scheunert, and E.~A. Jorswieck, ``An efficient branch-and-bound
  algorithm for compute-and-forward,'' in \emph{Proc. of the 23rd IEEE
  International Symposium on Personal, Indoor and Mobile Radio Communications
  (PIMRC'12)}, 2012.

\bibitem{Wei2012}
L.~Wei and W.~Chen, ``Compute-and-forward network coding design over
  multi-source multi-relay channels,'' \emph{IEEE Transactions on Wireless
  Communications}, vol.~11, no.~9, pp. 3348--3357, 2012.

\bibitem{Mochaourab2011}
R.~Mochaourab and E.~A. Jorswieck, ``Optimal beamforming in interference
  networks with perfect local channel information,'' \emph{IEEE Transactions on
  Signal Processing}, vol.~59, no.~3, pp. 1128--1141, 2011.

\bibitem{Wilson2010}
M.~P. Wilson, K.~Narayanan, H.~D. Pfister, and A.~Sprintson, ``Joint physical
  layer coding and network coding for bidirectional relaying,'' \emph{IEEE
  Transactions on Information Theory}, vol.~56, no.~11, pp. 5641--5654, Nov.
  2010.

\bibitem{Zhang2009}
S.~Zhang and S.-C. Liew, ``Channel coding and decoding in a relay system
  operated with physical-layer network coding,'' \emph{IEEE Journal on Selected
  Areas in Communications}, vol.~27, no.~5, pp. 788--796, Jun. 2009.

\bibitem{Popovski2007}
P.~Popovski and H.~Yomo, ``Physical network coding in two-way wireless relay
  channels,'' in \emph{IEEE International Conference on Communications (ICC)},
  Jun. 2007, pp. 707--712.

\bibitem{ElGamal2013}
A.~El~Gamal, O.~O. Koyluoglu, M.~Youssef, and H.~El~Gamal, ``Achievable secrecy
  rate regions for the two-way wiretap channel,'' \emph{IEEE Transactions on
  Information Theory}, vol.~59, no.~12, pp. 8099--8114, Dec. 2013.

\bibitem{Bloch2013}
M.~R. Bloch and J.~N. Laneman, ``Strong secrecy from channel resolvability,''
  \emph{IEEE Transactions on Information Theory}, vol.~59, no.~12, pp.
  8077--8098, Dec. 2013.

\bibitem{He2014}
X.~He and A.~Yener, ``Providing secrecy with structured codes: Two-user
  {G}aussian channels,'' \emph{IEEE Transactions on Information Theory},
  vol.~60, no.~4, pp. 2121--2138, Apr. 2014.

\bibitem{Telatar1999}
E.~Telatar, ``Capacity of multi-antenna {G}aussian channels,'' \emph{European
  Transactions on Telecommunications}, vol.~10, no.~6, pp. 585--595, 1999.

\end{thebibliography}

\end{document}